\documentclass[a4paper,11pt,onecolumn]{article} % Aptara syntax

\usepackage{times,epsfig}
\usepackage[ruled,vlined,linesnumbered,norelsize]{algorithm2e}
\usepackage{graphicx}
\usepackage{array}
\usepackage{centernot}
\usepackage{mdwmath}
\usepackage{mdwtab}
\usepackage{url}
\usepackage{icomma}

\usepackage{amsthm}

\newtheorem{mylemma}{Lemma}
\newtheorem{mydefinition}{Definition}
\newtheorem{innercustomthm}{Theorem}
\newenvironment{customthm}[1]
  {\renewcommand\theinnercustomthm{#1}\innercustomthm}
  {\endinnercustomthm}
\usepackage{color}
\usepackage{enumitem}
\usepackage{listings}
\newcommand{\MQ}{MQ\xspace}
\newcommand{\MQLF}{MQ$_{LF}$\xspace}
\newcommand{\TUMLSC}{TuML$_{SC}$\xspace}
\newcommand{\TUMLMC}{TuML$_{MC}$\xspace}
\newcommand{\WIML}{WiML\xspace}

\newcommand{\winset}{winset\xspace}
\newcommand{\winsets}{winsets\xspace}
\newcommand{\tuplist}{T-Gate\xspace}
\newcommand{\winlist}{W-Hive\xspace}
\newcommand{\smergeable}{\textit{ready}}
\newcommand{\mergets}{merge_{ts}\xspace}

\newcommand{\addX}{\textit{Add}\xspace}
\newcommand{\smergeX}{\textit{S-Merge}\xspace}
\newcommand{\updateX}{\textit{Update}\xspace}
\newcommand{\outputX}{\textit{Output}\xspace}

\newcommand{\lref}[2]{L\ref{#2}}

\usepackage{etoolbox}
\apptocmd{\thebibliography}{%
  \setlength{\itemsep}{-1pt}%
  \setlength{\parskip}{0pt}%
}{}{}

\begin{document}

\title{Efficient data streaming multiway aggregation through concurrent algorithmic designs and new abstract data types\thanks{
A brief announcement about parts of this work has been accepted at the 26th ACM Symposium on Parallelism in Algorithms and Architectures (SPAA), 2014~\cite{cederman2014brief}.
The research leading to these results has been partially supported by  the European Union Seventh Framework Programme (FP7/2007-2013) through the EXCESS Project (www.excess-project.eu) under grant agreement 611183, through the SysSec Project, under grant agreement 257007, through the FP7-SEC-285477-CRISALIS project, by the collaboration framework of Chalmers Energy Area of Advance and by the Chalmers Center for E-science.}}

\author{Vincenzo Gulisano, Yiannis Nikolakopoulos, Daniel Cederman \\ Marina Papatriantafilou and Philippas Tsigas \\ \small{\url{{vinmas,ioaniko,cederman,ptrianta,tsigas}@chalmers.se}}}
\date{}
%
%\begingroup
%\centering
%{\LARGE  \\[1.5em]
%\large Vincenzo Gulisano, Yiannis Nikolakopoulos, Daniel Cederman, Marina Papatriantafilou and Philippas Tsigas}\\[1em]
%Chalmers University of Technology
%\url{\{vinmas,ioaniko,cederman,ptrianta,tsigas\}@chalmers.se}
%\twocolumn
%\endgroup

%\markboth{V. Gulisano et al.}{Efficient streaming aggregation through concurrent algorithmic designs and new ADTs}

%\title{Efficient data streaming multiway aggregation through concurrent algorithmic designs and new abstract data types}
%\author{Vincenzo Gulisano
%\affil{Chalmers University of Technology}
%Yiannis Nikolakopoulos
%\affil{Chalmers University of Technology}
%Daniel Cederman
%\affil{Chalmers University of Technology}
%Marina Papatriantafilou
%\affil{Chalmers University of Technology}
%Philippas Tsigas
%\affil{Chalmers University of Technology}}

\maketitle

\begin{abstract}

Data streaming relies on continuous queries to process unbounded streams of data in a real-time fashion.
It is commonly demanding in computation capacity, given that the relevant applications involve very large volumes of data. Data structures act as articulation points and maintain the state of data streaming operators, potentially supporting high parallelism and balancing the work between them. Prompted by this fact, in this work we study and analyze parallelization needs of these articulation points, focusing on the problem of streaming multiway aggregation, where large data volumes are received from multiple input streams. The analysis of the parallelization needs, as well as of the use and limitations of existing aggregate designs and their data structures, leads us to identify needs for proper  shared objects that can achieve low-latency and high-throughput multiway aggregation. We present the requirements of such objects as abstract data types and we provide efficient lock-free linearizable algorithmic implementations of them, along with new multiway aggregate algorithmic designs that leverage them, supporting both deterministic order-sensitive and order-insensitive aggregate functions. Furthermore, we point out future directions that open through these contributions. The paper includes an extensive experimental study, based on a variety of aggregation continuous queries on two large datasets extracted from SoundCloud, a music social network, and from a Smart Grid network. In all the experiments, the proposed data structures and the enhanced aggregate operators improved the processing performance significantly, up to one order of magnitude, in terms of both throughput and latency, over the commonly-used techniques based on queues.
\end{abstract}

%\category{H.2.4}{Database Management}{Systems - Query Processing}

%\terms{Algorithms; Design; Experimentation; Performance; Theory}

%\keywords{data streaming, data structures, lock-free synchronization}

%\acmformat{Vincenzo Gulisano, Yiannis Nikolakopoulos, Daniel Cederman, Marina Papatriantafilou and Philippas Tsigas, 2014. Efficient data streaming multiway aggregation through concurrent algorithmic designs and new abstract data types.}

%\thanks{
%A brief announcement about parts of this work has been accepted at the 26th ACM Symposium on Parallelism in Algorithms and Architectures (SPAA), 2014.

%The research leading to these results has been partially supported by  the European Union Seventh Framework Programme (FP7/2007-2013) through the EXCESS Project (www.excess-project.eu) under grant agreement 611183, through the SysSec Project, under grant agreement 257007, through the FP7-SEC-285477-CRISALIS project, by the collaboration framework of Chalmers Energy Area of Advance and by the Chalmers Center for E-science.

%Author's mail address: \{vinmas,ioaniko,cederman,ptrianta,tsigas\}@chalmers.se.
%}

\section{Introduction}\label{sec:introduction}

For data intensive computing that can support continuous complex analysis of large volumes data, the data streaming processing paradigm emerged as a more appropriate alternative to the traditional ``store-then-process'' one.
As emphasized in~\cite{gedik2009celljoin}, the low-latency and high-throughput requirements of such continuous real-time complex processing of increasingly large data volumes make parallelism a necessity.

In data streaming~\cite{stonebraker20058,abadi2003aurora,Dobra:2002}, \textit{continuous queries}, defined as Directed Acyclic Graphs (DAGs) of interconnected operators, are executed by Stream Processing Engines (SPEs) that process incoming data in a real-time fashion, producing results on an on-going basis.
A good portion of the research has so far focused on leveraging the processing capacity of clusters of nodes and originally centralized SPEs~\cite{abadi2003aurora} evolved rapidly to distributed~\cite{balazinska2008fault} and parallel~\cite{gulisano2012streamcloud,loesing2012stormy} ones.
At the same time, research has also focused on leveraging multi-core CPUs and GPUs architectures, as discussed in \cite{cugola2012low,schneider2009elastic,schneidert2010evaluation}

A parallel data streaming application can be seen as a pipeline where data is continuously produced, processed and consumed. In a parallel environment the underlying data structures should provide the means for organizing the data so that the communication  and the work imbalance between the concurrent threads performing the computation are minimized while the pipeline parallelism is maximized.
Finding the appropriate data structures that fit the needs of an application in a concurrent environment is a key research issue \cite{Shavit:2011,Michael:2013}.
Defining and providing the data structures that meet the needs of concurrent data streaming applications is a rich issue not addressed in the literature; this is all the more important, given the high performance demands of the relevant applications. 

By shedding light on the data structures, we identify new key challenges to improve data streaming aggregation, one of the most common and throughput-demanding monitoring application~\cite{admt-2012-001}.
In particular, we focus on multiway aggregation, where big volumes of data received from multiple input streams must be merged and sorted in order to be processed deterministically~\cite{gulisano2012streamcloud}.
Sample application scenarios include monitoring applications in the context of social media, where information could be aggregated to study trends, or in the context of real-time pricing applications in Smart Grids, or in the context of adaptive traffic systems.

\paragraph*{Contributions}
We study data structures as articulation points between pipeline stages of streaming aggregation.
The shared access to the data by the collaborating threads defines new synchronization needs that can be integrated in the functionality provided by the shared data structures.
By studying the use and limitations of existing aggregate designs and the data structures they use, we motivate the need for shared data objects appropriate for streaming aggregation.
We propose two types of such objects (\textit{\tuplist} and \textit{\winlist}) and their concurrent and lock-free algorithmic implementations, upon which we build three enhanced multiway aggregate operators that balance the work among concurrent threads and outperform existing implementations in both order-sensitive and order-insensitive functions.
We provide an extensive study using two large datasets extracted from the SoundCloud\footnote{\url{https://soundcloud.com/}} social media and from a Smart Grid network.
For both datasets the enhanced aggregation resulted in large improvements, up to one order of magnitude, both
in terms of processing throughput and latency.

Our contributions open up space for new research questions for the role of concurrent data structures in parallel data streaming and are expected to influence significantly the design and implementations of parallel SPEs.

The paper is organized as follows.
Section~\ref{sec:datastreamingandmultiwayaggregation} introduces the data streaming processing paradigm and the multiway aggregate operator.
Section~\ref{sec:rethinking} presents the state of the art implementation of data streaming multiway aggregation and, by rethinking parallelism in this context, discusses how its efficiency can be enhanced by means of concurrent data structures.
Sections~\ref{sec:newabstractdatatypesandaggregatedesigns} and~\ref{sec:aggregatesanddatastructuresimplementations} present a detailed overview of the algorithmic design and implementation of the enhanced operators and data structures that we propose.
In Section~\ref{sec:correctness} we show the liveness and safety properties, namely lock-freedom and linearizability, of the proposed operators and data structures implementations.
Section~\ref{sec:evaluation} presents the experimental evaluation.
We discuss related work in Section~\ref{sec:relatedwork} and conclude in Section~\ref{sec:conclusions}.

\section{Data Streaming and Multiway Aggregation}\label{sec:datastreamingandmultiwayaggregation}

A stream is defined as an unbounded sequence of tuples $t_0,t_1,\ldots$ sharing the same schema composed by attributes $\langle ts,\ A_1,\ldots,A_n \rangle$.
Given a tuple $t$, attribute $t.ts$ represents its creation timestamp while $A_1,\ldots,A_n$ are application-related attributes.
Following the data streaming literature (e.g., \cite{balazinska2008fault,gulisano2012streamcloud}), we assume that each stream contains timestamp-sorted tuples.

\textit{Continuous queries} (or simply queries) are defined as DAGs of {\em operators} that consume and produce tuples.
Operators are distinguished into \textit{stateless} or \textit{stateful}, depending on whether they keep any state that evolves with the tuples being processed.
Due to the unbounded nature of streams, stateful operations are computed over a \textit{sliding window}, defined by parameters \textit{size} and \textit{advance}.
Sliding windows can be time-based (e.g., to group tuples received during periods of $5$ minutes every $2$ minutes) or tuple-based (e.g., to group the last $10$ received tuples every $3$ incoming tuples).
We focus in this paper on time-based sliding windows (or simply windows).
We use POSIX notation\footnote{Defined as the number of seconds elapsed since Thursday, 1 January 1970} to specify the periods covered by a window and assume all windows start at time $0$. That is, a window with size and advance of $10$ and $2$ seconds, respectively, will cover periods $[0,10)$, $[2,12)$, $[4,14)$, and so on.

The multiway aggregate operator consumes an arbitrary number of input streams and is defined by its window's \textit{size} and \textit{advance}, by a function $F$ applied to the tuples and by an optional \textit{group-by} parameter $K$ (a subset of the input tuple's attributes, also referred to as the tuple's \emph{key}).
Functions $F$ can be {\em order-sensitive} (e.g., forward only the first received tuple) or {\em order-insensitive} (e.g., count the number of tuples) with respect to the processing order of the tuples that contribute to the same window.
If $K$ is defined, function $F$ is computed for each distinct value of the \textit{group-by} parameter.
In this case, the operator keeps separate windows not only for different time intervals, but also for different values of $K$.

We define a \textit{\winset} as the set of windows covering the same time interval for different values of $K$.
As an example, suppose an aggregate operator consumes tuples composed by attributes $\langle ts,\ meter,\ consumption \rangle$ (each referring to the consumption reported by a meter at time $ts$) and computes the average consumption for $K=meter$ and for a window with size and advance of $10$ and $2$ seconds, respectively. In this case, \winset{}s will hold the windows covering interval $[0,10)$ for each distinct meter, the windows covering interval $[2,12)$ for each distinct meter, and so on.

Scenarios such as parallel-distributed SPEs~\cite{gulisano2012streamcloud,balkesen2013adaptive} and replica-based fault tolerant SPEs~\cite{balazinska2008fault}, demand for deterministic aggregation of input tuples.
Processing of a multiway aggregation is deterministic if tuples are processed in timestamp order (when F is order-sensitive) or if all the tuples contributing to the same window are processed before producing the result (when F is order-insensitive).
To this end, a parallel execution of a multiway aggregation must ensure that tuples are not simply processed in the order they are received \cite{gulisano2012streamcloud} (i.e., input tuples from different input streams are not arbitrarily interleaved).
To ensure deterministic processing, tuples from multiple input streams need to be \emph{merged} into one sequence and \emph{sorted} in timestamp order \cite{gulisano2012streamcloud}, an operation we refer to as \smergeX.
We say that a tuple is \smergeable{} to be processed if at least one tuple with an equal or higher timestamp has been received at each input stream (we refer the reader to Section~\ref{sec:rethinking} for an example focusing on the processing of \smergeable{} tuples).
\begin{mydefinition}\label{def:ready}
Let $t^j_i$ be the $i$-th tuple received from input $j$.  $t^j_i$ is \smergeable{} to be processed if $t^j_i.ts \leq \mergets,$ where $\mergets = min_k\{max_l(t^k_l.ts)\}$ is the
minimum among the latest (over $l$) tuple timestamps received from every input $k$.
\end{mydefinition}
Thus, we can formalize deterministic aggregate operators for both order sensitive and insensitive functions as follows:
\begin{mydefinition}\label{def:determinism}
An aggregate operator implementation is deterministic if in every execution it is guaranteed that the output tuples are computed from \smergeable{} tuples and they are processed in the order imposed by the respective aggregate function $F$.
\end{mydefinition}

\newpage
\section{Rethinking aggregation's parallelism: the role of data structures}\label{sec:rethinking}

The multiway aggregate operator is composed by four main stages:
\begin{enumerate}[noitemsep]
 \item \addX{}: fetching incoming tuples from each input stream,
 \item \smergeX: merging and sorting of input streams' tuples,
 \item \updateX: updating of the windows a tuple contributes to, and
 \item \outputX: forwarding of output tuples.
\end{enumerate}

Figure~\ref{fig:cqexample} presents a sample multiway aggregation query used in a Smart Grid application to count the number of power outages reported by each meter over a sliding window with size and advance of 3 and 2 time units, respectively.
A mesh network of Smart Meters (SMs) forwards such alarms to a set of Data Concentrators (DCs) \cite{gungor2011smart}, which in turn produce the timestamp sorted input streams.
Notice that, being a mesh network, messages generated by the same SM could be forwarded by distinct DCs.
The input tuples' schema $\langle ts, meter \rangle$ specifies the time $ts$ at which the alarm forwarded by a given $meter$ has been received by a DC.
Tuples produced by $A$ are composed by attributes $\langle ts, meter, \#alarms \rangle$ and specify the number of alarms generated by each meter for the window starting at time $ts$.
In the example, the aggregate operator has two input streams.
Input tuple $\langle 4,SM_1 \rangle$ is received at the first input stream, while tuple $\langle 6,SM_1 \rangle$ is received at the second input stream.
The figure presents the different steps performed by the operator for a given initial state.
Notice that, given Definition~\ref{def:ready}, tuple $\langle 4,SM_1 \rangle$ is \smergeable{} to be processed  while $\langle 6,SM_1 \rangle$ is not.

\begin{figure}
\centering
\includegraphics*[width=1\linewidth]{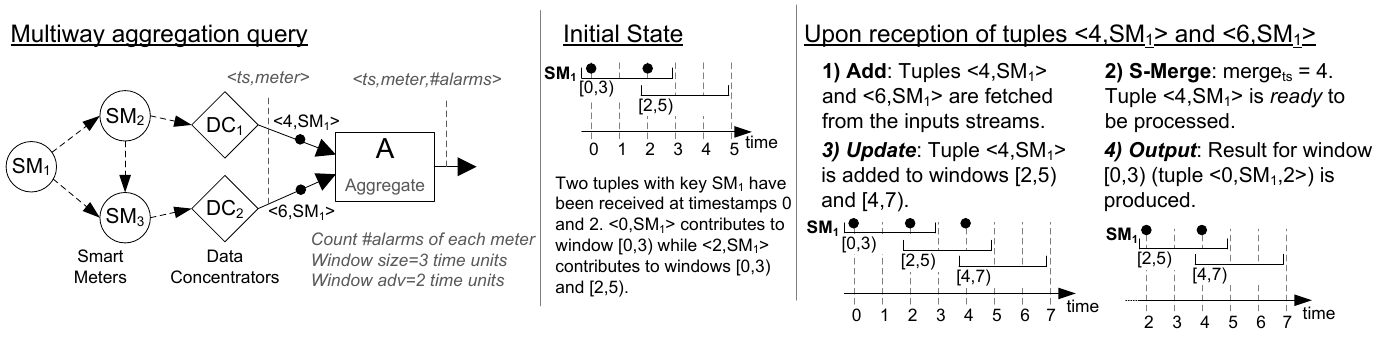}
  \caption{Sample query to count the power outages reported by smart meters and sample execution of the aggregate operator for a sliding window with size and advance of 3 and 2 time units, respectively.}
  \label{fig:cqexample}
\end{figure}

\paragraph*{State of the art}
Widely used SPEs such as Borealis \cite{balazinska2008fault} or StreamCloud \cite{gulisano2012streamcloud} perform multiway aggregation by relying on per-input queues to store incoming tuples.
Distinct threads, which we refer to as input threads $I_t$, insert tuples to such queues while a dedicated output thread $O_t$ processes them.
Concurrent accesses are synchronized with the help of locks.
Figure~\ref{fig:implementations}a presents this design, which we refer to as Multi-Queue (\MQ{}).
The output thread $O_t$ peeks the first tuple in each queue to determine which one is \smergeable{} to be processed. The same thread is also responsible for the \updateX{} and \outputX{} operations.
Since $O_t$ is the only thread in charge of updating windows, no locking mechanism is required to access the \textit{\winsets}, usually implemented as hash tables to easily support arbitrary numbers of windows and to locate them quickly given the group-by parameter K.

\paragraph*{Parallelization challenges}\label{par:parallelizationchallenges}
In existing implementations, \smergeX{} usually relies on simple sorting techniques, whose cost is linear to the number of inputs.
Examples include the \textit{Input Merger} operator \cite{gulisano2012streamcloud} or the \textit{SUnion} operator \cite{balazinska2008fault}.
In order to prevent such sorting techniques from becoming a bottleneck and allow for the processing of tuples coming from arbitrary number of input streams, the first challenge relies on the (1) parallelization of the \smergeX{} operation. It should be noticed that, an enhanced parallel sorting technique should still ensure correct pick-up of \smergeable{} tuples.
The second challenge relies on the (2) parallelization of the \updateX{} stage.
To guarantee deterministic processing (as discussed in Section~\ref{sec:datastreamingandmultiwayaggregation}), \updateX{} cannot be invoked in parallel on tuples contributing to the same window and sharing the same $K$ value (or when no group-by parameter is defined) for order-sensitive functions.
This restriction can be relaxed for order-insensitive functions, since the result of a window would not be affected by the order in which concurrent threads update it.
Since the aggregate operator defines a single output stream to which tuples are added in timestamp order, we do not take into account the parallelization of the Output function.

\paragraph*{Utilizing concurrent data structures}

A core challenge in the parallelization of the pipeline stages of multiway aggregation is the ``{\em balancing act}" \cite{Michael:2013} of maximizing their concurrency while ensuring consistency and correct synchronization.
To this end, the key-enablers that can address such challenge are the data structures, seen as articulation points between such stages, as well as their efficient algorithmic implementations.
Efficient implementations of data structures often employ fine-grain synchronization that can avoid the use of waiting or locking.
Lock-free data structures have been shown to increase applications' throughput and are part of the Java and \verb!C#! standard libraries.
The correctness of such implementations is commonly  shown through \emph{linearizability}~\cite{herlihyandwing}, which guarantees that, given a history of concurrent operations, there exists a sequential ordering of them, consistent with their real-time ordering and with the sequential semantics of the data structure.

In order to parallelize the \smergeX{} and \updateX{} stages, we first explored which existing concurrent data structures could be used to sort input tuples at insertion time.
In principle, tree-like data structures could provide concurrent logarithmic-time insertion operations.
The need for easily extracting such tuples in timestamp order would be better addressed by a concurrent skip list (e.g., the one proposed by~\cite{sundell_fast_2005}) due to its underlying list-like node structure of sorted elements with shortcuts allowing for fast insertion (cf. Sec.~\ref{subsec:skiplists} for more details on skip lists).
Nevertheless, a skip list would not differentiate between tuples that are \smergeable{} and tuples that are not.
Because of that, checking whether a tuple is \smergeable{} or not would still be penalized by a cost that is linear to the number of inputs, as is the case for the multi-queue implementations.
Based on this observation, while leveraging the skip list's multi-level shortcuts mechanism (allowing for a logarithmic find of the insert position in the list), we propose  new concurrent shared data object types that better fit the parallelization challenges proper of multiway aggregation.
We complement the qualitative estimation of the reasons that motivated the design and implementation of new concurrent data structures by comparing them with a lock-free skip list in Section~\ref{sec:evaluation}.

\section{New abstract data types and aggregate designs}
\label{sec:newabstractdatatypesandaggregatedesigns}

\begin{table}[t]
\caption{Methods supported by the data structures}
\label{tab:api}
%\tbl{Methods supported by the data structures.\label{tab:api}}{
\scriptsize
\centering
\begin{tabular}{|p{2.7cm}p{3cm}|p{2.7cm}p{3cm}|}
\hline
{\bf \tuplist} & & {\bf \winlist} &\\
\hline
\texttt{insertTuple}(tuple, input) & Inserts a \emph{tuple} from \emph{input} stream in sorted order. & \texttt{updateWindows}(tuple, thread) & Updates the windows that the tuple contributes to. \\
\hline
\texttt{getNextReadyTuple}() & Returns (once and only once) the earliest \smergeable{} tuple (cf. definition~\ref{def:ready}). & \texttt{getNextWinSet}() & Returns (once and only once) the earliest \winset to which tuples do not contribute anymore. \\
\hline
\end{tabular}
%}
\end{table}

\begin{figure}
  \centering
  \includegraphics*[width=1\linewidth]{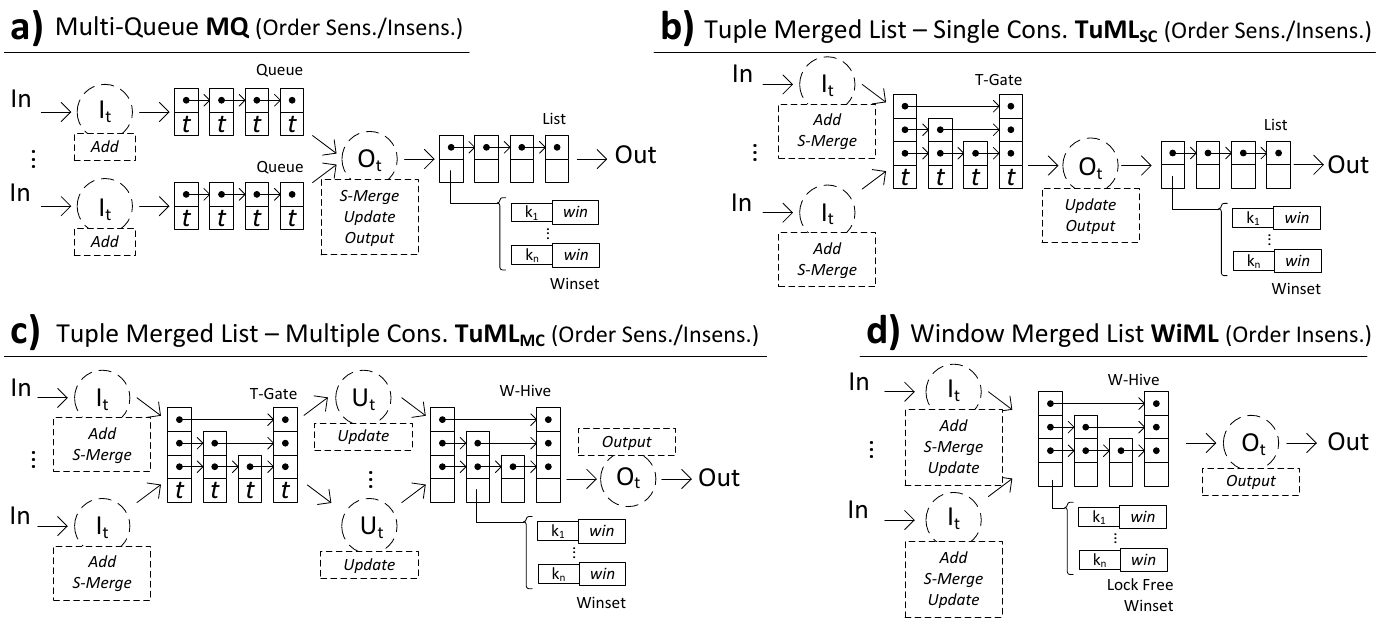}
  \caption{Overview of aggregate designs.}
  \label{fig:implementations}
\end{figure}

This section overviews our enhanced aggregate operators.
For all of them, one input thread per stream, $I_t$, fetches tuples from its respective input stream while a single output thread, $O_t$, forwards output tuples.
Figure~\ref{fig:implementations} presents the different designs and describes how operations are assigned to threads.
While presenting the different designs, we discuss the data structures needed to maintain tuples and \textit{\winsets}, and introduce our concurrent data structures and their APIs (Table~\ref{tab:api}), with the functionality of each method.

\paragraph*{Tuple Merged List - Single Consumer (\TUMLSC)}
This algorithmic design \\ (Fig.~\ref{fig:implementations}b) addresses the first parallelization challenge by performing both \addX{} and \smergeX{} in parallel.
\TUMLSC{} relies on the Tuple-Gate (\tuplist{}), a concurrent data structure whose API provides two methods, whose definitions are given in Table~\ref{tab:api}.
Method \texttt{insertTuple}(tuple, input) allows for tuples to be inserted, while being merged and sorted, by multiple input streams in parallel.
Method \texttt{getNext\allowbreak Ready\allowbreak Tuple}() guarantees that no tuple is returned for processing before it is \smergeable{}.
The T-Gate stores an ordered list of tuples.
By keeping track of the latest added tuple from each input stream, the method can quickly check if the first tuple in the list is \smergeable{}.
The output thread reads sorted and \smergeable{} tuples from the T-Gate and performs the \updateX{} and \outputX{} stages.

\paragraph*{Tuple Merged List - Multiple Consumer (\TUMLMC{})}

\TUMLMC (Fig.~\ref{fig:implementations}c) extends the \TUMLSC, addressing the second parallelization challenge by performing also the \updateX operation in parallel.
Multiple update threads, $U_t$, get \smergeable{} tuples from \tuplist{} concurrently by invoking \texttt{getNextReadyTuple}() and update the windows to which each tuple contributes to.
Thus, \winset{}s are now accessed and updated concurrently by the $U_t$ threads.
For managing the \winset{}s and synchronizing such access we introduce a second data structure, that we refer to as Window-Hive (\winlist{}).

As the \tuplist{} encapsulates the logic to differentiate between tuples that are \smergeable{} or not, the \winlist{} is able to differentiate between the \winset{}s to which incoming tuples are still contributing  and the ones whose results can be outputted.
It provides two methods: \texttt{updateWindows}(tuple, thread) allows for multiple threads to synchronize and safely create and update active \winset{}s while \texttt{getNext\allowbreak WinSet}() returns the earliest \winset{} no longer being updated by any thread.
This method is invoked by the output thread $O_t$, in charge of forwarding the operator's output tuples.
W-Hive uses similar techniques as the T-Gate to quickly find the right location of where to insert a new \winset{}.
To preserve the correctness of order-sensitive functions, each update thread is responsible for a distinct subset of the group-by parameter values $K$.
For this implementation, the number of update threads can be chosen by the user.

\paragraph*{Window Merged List (\WIML)}
This design (Fig.~\ref{fig:implementations}d) further enhances the parallelization of the aggregate's stages for order-insensitive functions.
Operations \addX{}, \smergeX{} and \updateX{} are performed in parallel by the $I_t$ threads.
Since \WIML{} is designed for order-insensitive functions, input tuples do not need to be sorted before being processed to update the windows they contribute to. The required synchronization needed to ensure that output tuples for a given \winset{} are outputted only after all its contributing tuples have been processed, is managed by the \texttt{getNextWinSet}() method provided by the \winlist{}.

\section{Aggregates and data structures implementations}\label{sec:aggregatesanddatastructuresimplementations}
In this section we present in detail the aggregate algorithmic implementations and their supporting data structures.
The pseudocode for the baseline aggregate implementation can be found in Algorithm~\ref{lst:sequential}, 
while the enhanced implementations in Algorithms~\ref{lst:lfcore} and~\ref{lst:combined}.
The supporting data structures are presented in Algorithms~\ref{lst:tuplesl} and~\ref{lst:winsl}.
Methods' names have been chosen according to which of the four main aggregate's stages specified in Section~\ref{sec:rethinking} they implement.
In the following, the group-by value of each tuple is accessed as \texttt{tuple.key}. If no group-by parameter is defined, it is safe to assume all tuples will refer to the same \texttt{key} value (e.g., \texttt{null}).

\subsection{Preliminaries}\label{subsec:skiplists}
A skip list~\cite{pugh_skip_1990} is a data structure that maintains elements in an ordered list and supports probabilistically logarithmic search, insertion and deletion operations.
Essentially, a skip list can be viewed as a traditional linked list where each node, besides the 
usual pointer connecting to the next element, has a tower of several pointers that shortcut over the next elements and connect to nodes later in the ordered list. 
The height of the nodes is randomly distributed so that $50\%$ of the nodes have height 1, $25\%$ of them have height 2 and so on.
Thus, the higher the level traversed, the sparser the links are and more nodes are skipped.
The basic search routine for a key $k$, is to traverse from the highest level shifting to a lower one every time the current node's key is greater than $k$.

One of the main benefits of skip lists over standard tree like data structures is that regardless of the data and operation distribution there is no need for rebalancing. 
This has made it a good candidate for parallel and concurrent implementations, as the one by~\cite{sundell_fast_2005}, since rebalancing will typically require expensive synchronization in tree-based implementations.

\subsection{Common components}
\lstset{emph={processTuple,produceOutTuple},emphstyle=\bfseries}
\begin{algorithm}

\begin{lstlisting}[language=Java, name=code]
interface Window
	void processTuple(tuple)	// update variables @\label{line:windowadd}@
	Tuple produceOutTuple() // produce output tuple @\label{line:genoutput}@

class SumWindow : Window
	int sum = 0
	void processTuple(tuple)
		sum += tuple.$A_i$
	Tuple produceOutTuple()
		return Tuple(sum)
\end{lstlisting}
\caption{Generic window interface and concrete implementation of a window that sums the value of attribute $A_i$.}
\label{lst:window}

\end{algorithm}
The base component of the aggregate operator is the \texttt{Window}.
It represents a time interval and provides functionality to aggregate the tuples that contribute to it.
Algorithm~\ref{lst:window} shows the window interface and the implementation of a sample \texttt{sum} aggregation.

As discussed in Section~\ref{sec:rethinking}, the \textit{\winset} can be implemented as a hash table to easily support an arbitrary number of windows and to locate them quickly given the tuple's group-by parameter K.
In the \WIML and \TUMLMC implementations, the \winset is accessed by multiple threads concurrently.
In data streaming applications the domain of keys for the group-by parameters are typically known in advance. 
Thus, the use of a closed addressing hash table is appropriate; specifically, in our algorithmic implementation we are using the lock-free, linearizable concurrent hash table by Michael~\cite{michael2002}, mainly due to its implementation simplicity.
Alternatively, open addressing schemes like lock-free cuckoo hashing~\cite{nguyen_lock-free_2014} can be used.
Furthermore, the hash table (\textit{\winset}) is a building a block in another lock-free and linearizable data structure (\winlist).
Therefore, we avoided designs based on blocking implementations~\cite{herlihy_hopscotch_2008} or
explicit hardware support~\cite{li_algorithmic_2014}.
Shun and Blelloch~\cite{shun_phase-concurrent_2014} present a high performing phase-concurrent hash table. In this model only operations of the same type proceed concurrently, which is a limitation in the winset use-case since insert and find operations may be concurrent.
For the \MQ{} and \TUMLSC{} implementations, since the access to the \winset is sequential, a sequential implementation of a hash table is sufficient.

\subsection{Baseline implementations - \MQ}
\lstset{emph={processTuple,Add,insert,SMergeUpdateOutput,produceOutTuple},emphstyle=\bfseries}
\begin{algorithm}
\begin{lstlisting}[name=code]
Add(tuple, input) // One thread per input
	queue$_{input}$.enqueue(tuple) @\label{line:queueadd}@

SMergeUpdateOutput() // One thread
	if($\exists$i : queue$_i$.isEmpty()) return @\label{line:fromall}@
	input = v:($\forall$i:queue$_v$.peek().ts$\leq$queue$_i$.peek().ts) @\label{line:lowestbegin}@
	tuple = queue$_{input}$.dequeue() @\label{line:lowestend}@
	upout(tuple)

upout(tuple) @\label{line:upout}@ 
	windowTSs = getTargetWindowTSs(tuple)
	while(windowlist.first().ts<windowsTSs.first()) @\label{line:readytoforwardbegin}@
		winset = windowlist.removeFirst()
		for (window : winset)
			forward(window.produceOutTuple()) @\label{line:readytoforwardend}@ // See @\lref{lst:window}{line:genoutput}@
	for (wts : windowTSs)
		if(!windowlist.contains(wts)) @\label{line:nohash}@
			windowlist.insert(wts, new WinSet(wts))
		win = windowlist.find(wts).find(tuple.key)
		if (win == null) @\label{line:nowin}@
			win = new Window()
			windowlist.find(wts).put(tuple.key, win)
		win.processTuple(tuple) @\label{line:addtuplemq}@ // See @\lref{lst:window}{line:windowadd}@
\end{lstlisting}
%\caption{MQ and MQ$_{LF}$}
\caption{MQ}
\label{lst:sequential}
\end{algorithm}
This baseline implementation is based on the one used in SPEs such as Borealis \cite{balazinska2008fault} or StreamCloud \cite{gulisano2012streamcloud}.
The multi-queue design consists of two main methods (see Algorithm~\ref{lst:sequential}).
The \texttt{Add} method is used to deliver tuples to the aggregate and placing them in their respective input queue (L\ref{line:queueadd}).
The queues are protected by a lock to allow concurrent access.

The main work is performed by the second method, \texttt{SMerge\allowbreak Update\allowbreak Output}.
It checks all the queues to make sure a tuple has been received from each input (L\ref{line:fromall}).
It then reads the tuple with the lowest timestamp among the inputs (L\ref{line:lowestbegin}-\ref{line:lowestend}).
This guarantees that all tuples will be read in timestamp order.

The currently active \winset{}s are stored in a linked list.
The method \\ \texttt{getTargetWindowTSs} creates a list, \texttt{windowTSs}, of the starting timestamps of the windows that the tuple contributes to.
If the starting timestamp of a \winset in the window list is lower than the earliest timestamp in \texttt{windowTSs},
the aggregated results of the former can be outputted (L\ref{line:readytoforwardbegin}-\ref{line:readytoforwardend}).
This is safe since all future tuples will have an equal or higher timestamp and will not contribute to the \winset.
If the new tuple contributes to a time interval that does not have a corresponding \winset yet, the \winset is created and added to the list (L\ref{line:nohash}).
If the window does not exist for the tuple's key, it is also created (L\ref{line:nowin}).
Finally, the window processes the tuple (L\ref{line:addtuplemq}).

Figure~\ref{fig:example_mq} presents how stages \addX{}, \smergeX{}, \updateX{} and \outputX{} (and their respective code lines) are distributed to threads $I_t$ and $O_t$ for the \MQ{} implementation (stages assigned to $I_t$ and $O_t$ threads are colored in blue and red, respectively).

\begin{figure}[ht!]
  \centering
  \includegraphics*[width=1\linewidth]{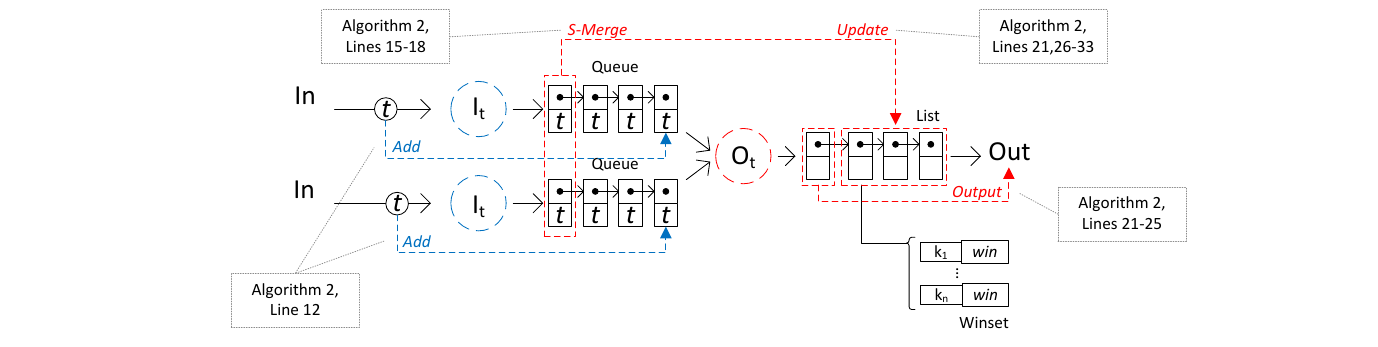}
  \caption{Visual representation of how stages (and their respective code lines) are distributed to threads for the \MQ{} implementation.}
  \label{fig:example_mq}
\end{figure}

\subsection{\TUMLSC{} and \TUMLMC{}}
\lstset{emph={UpdateOutput,getNextWinSet,processTuple,AddSMerge,insertTuple,getNextReadyTuple,updateWindows,Update,Output,processedTS,produceOutTuple},emphstyle=\bfseries}
\begin{algorithm}
%\textbf{\TUMLSC{} and \TUMLMC{}}
\begin{lstlisting}[name=code, breaklines=true]
AddSMerge(tuple, input) // One thread per input
	tgate.insertTuple(tuple, input) // See @\lref{lst:tuplesl}{line:tladdop}@

UpdateOutput() @\label{line:updateoutput}@ // @\TUMLSC@ only
	tuple = tgate.getNextReadyTuple() // See @\lref{lst:tuplesl}{line:tlreadop}@ @\label{line:updateoutputread}@
	windowTSs = getTargetWindowTSs(tuple)
	// Produce results for windows no longer updated
	while(windowlist.first().ts<windowsTSs.first())
		winset = windowlist.removeFirst()
		for (window : winset)
			forward(window.produceOutTuple())
	// Update windows
	for (wts : windowTSs)
		if(!windowlist.contains(wts))
			windowlist.insert(wts, new WinSet(wts))
		win = windowlist.find(wts).find(tuple.key)
		if (win == null) 
			win = new Window()
			windowlist.find(wts).put(tuple.key, win)
		win.processTuple(tuple)

Update() // Multiple threads - @\TUMLMC@ only
	tuple = tgate.getNextReadyTuple() // See @\lref{lst:tuplesl}{line:tlreadop}@
	if (tuple == null) return
	if($\lnot$tuple.hashToThread(threadid)) return
	whive.updateWindows(tuple) // See @\lref{lst:winsl}{line:wladdop}@

Output() // One thread - @\TUMLMC@ only
	winset = whive.getNextWinSet() // See @\lref{lst:winsl}{line:wlreadop}@
	if(winset == null) return
	for(window : winset)
		// Output the result of the window
		forward(window.produceOutTuple())
\end{lstlisting}
\caption{\TUMLSC, \TUMLMC{}}
\label{lst:lfcore}
\end{algorithm} 
These aggregate designs rely on the \tuplist data structure (API in Table~\ref{tab:api} and further description in Section~\ref{sec:tuplist}).
The \tuplist is used to pre-sort all arriving tuples and merge them into one stream.
In contrast with the \MQ{} implementation, the \smergeX{} operation is now executed at the first stage in the pipeline.

\TUMLSC uses a single thread to read the sorted tuples from the \tuplist, update the windows, and output the aggregated results. This is done using the method \texttt{UpdateOutput}, which shares much functionality with the MQ design (L\ref{line:updateoutput}).
\TUMLMC allows multiple threads to read from the \tuplist and update the windows in parallel.
This requires support for concurrent handling of the \winset{}s. The \winlist (API in Table~\ref{tab:api} and further description in Section~\ref{sec:windowlist}) is used to provide lock-free \winset management.
If the aggregate opertor's function is order-sensitive (e.g., forward only the first received tuple), tuples contributing to the same window cannot be processed in parallel by multiple threads. Hence, a hash function based on the group-by attribute is used to assign input tuples to existing threads.

\subsection{\tuplist}
\label{sec:tuplist}
\lstset{emph={processed, insertTuple,getNextReadyTuple,processedTS,initializeTGate,levelinsert},emphstyle=\bfseries}
\begin{algorithm}
\begin{lstlisting}[name=code, breaklines=true]
Node head, update[maxlevels] // Thread local variables; maxlevels is a constant parameter

def Node
	Node next[maxlevels]
	Tuple tuple
	int input

initializeTGate() @\label{line:tgate_init}@
  tail = new Node()
  tmp = new Node()       // tmp is the temporary head
  for (i=0 to maxlevels-1) // all levels point to tail
    tmp.next$_i$ = tail 
  for (i in input ids)
    insertTuple(new Tuple(), i) //insert one dummy tuple per input

getNextReadyTuple() @\label{line:tlreadop}@
	next = head.next$_0$
	if(next$\neq$tail $\land$ written$_{next.input}\neq$next.tuple) @\label{lp:tgate_getNextReady}@
		head = next
		return next.tuple
	return null

insertTuple(tuple, input) @\label{line:tladdop}@
	nodeheight = getLevelHeight() @\label{line:getlevels}@
	newnode = new Node(tuple, input)
	curnode = update$_{maxlevels-1}$ @\label{line:lastseen}@
	for(i=maxlevels-1 downto 0) @\label{line:updatestart}@
		next = curnode.next$_i$
		while(next$\neq$tail $\land$ next.ts<tuple.ts)
			curnode = next
			next = curnode.next$_i$
		update$_i$ = curnode @\label{line:updateend}@
	for(i=0 to nodeheight)
		levelinsert(update$_i$, newnode, tuple.ts, i) @\label{line:levelins}@
	written$_{input}$ = newnode @\label{line:written}@
	
levelinsert(priornode, newnode, ts, level)
	while(true)
		next = priornode.next$_{level}$
		if(next==tail $\lor$ next.ts>ts)
			newnode.next$_{level}$ = next
			if(CAS(priornode.next$_{level}$, next, newnode)) break @\label{line:casLevelInsert}@
		else fromNode = next
\end{lstlisting}
\caption{\tuplist}
\label{lst:tuplesl}
\end{algorithm}
The \tuplist data structure (see Algorithm~\ref{lst:tuplesl}) maintains a merged, timestamp ordered
list of the tuples coming from the input streams.

The \texttt{insertTuple} method inserts a tuple at its correct position in the list, given its timestamp.
First the number of shortcut levels is decided (L\ref{line:getlevels}), according to the standard skip list distribution~\cite{sundell_fast_2005}.
Each thread keeps in the \texttt{update} array a pointer to the last accessed node in each level (L\ref{line:lastseen}).
Since all new tuples added by the same thread will have an equal or higher timestamp than the last inserted one,
this lowers the number of nodes that a thread has to examine.

Once all the levels are searched,the node is inserted on each level it should be part of with the use of the \texttt{levelinsert} helper method (L\ref{line:levelins}).
This method verifies that the conditions for the prior node in each level still apply, otherwise (in case some newer node has been inserted in between) it traverses the current level of the list until the right position is found.
The node is then inserted by using the compare and swap (CAS) atomic instruction. In case of failure, i.e. when another thread achieves an insertion at the same place, the loop retries the search.
When the node has been inserted, the \texttt{written} array is updated to hold a reference to the new node (L\ref{line:written}).
The index into the array is the input stream id. This is done to make sure a tuple is not read until we have received
a new tuple with a higher or equal timestamp from all the other input streams.

The \texttt{getNextReadyTuple} method traverses the lowest level of the list to return tuples in timestamp order. 
Each thread keeps its local head pointer having its own handle to the list and advances this pointer in each successful call.
A tuple pointed by the head can be returned if it is not the last one added by any input stream (the latter ensures that if a tuple is returned, it is indeed \smergeable{}).
It is useful to point out that in the case of just one tuple per input stream being present in the data structure, according to Definition~\ref{def:ready}, the tuple with the smallest timestamp is \smergeable.
However, the presented implementation will not return this tuple until another one with higher timestamp arrives from the same input stream.
This is done for implementation simplicity, since it does not compromise the correctness according to Def.~\ref{def:ready}, and does not affect the high input rate scenarios which we focus in this paper.

Finally, during the initialization of the data structure (L\ref{line:tgate_init}), one dummy tuple per input is inserted to ensure the correct semantics of the \texttt{getNextReadyTuple} are preserved until all input streams start delivering tuples.

Nodes can be freed when they are no longer accessible (directly or indirectly) from the thread local \texttt{head} and \texttt{update$_{maxlevels-1}$}.
For this reason, several memory reclamation techniques such as hazard pointers can be applied~\cite{michael_hazard_2004,sundell_fast_2005},
while also garbage collection can be exploited.
In the Java based implementation of our prototype that is evaluated in Section~\ref{sec:evaluation}, we rely on the default garbage collector.

\subsection{\winlist}
\label{sec:windowlist}
\lstset{emph={processTuple,char,Add,SMergeUpdateOutput,AddSMerge,updateWindows,getNextWinSet,peek,AddSMergeUpdate,Update,Output,processedTS,produceOutTuple},emphstyle=\bfseries}
\begin{algorithm}
\begin{lstlisting}[name=code, breaklines=true]
Node readhead, inserthead, tail

def Node
	Node next[maxlevels]
	Timestamp ts
	WinSet winset

getNextWinSet() @\label{line:wlreadop}@
	if(readhead.next$_0$==tail) return null @\label{line:lp1_whive}@
	if(readhead.next$_0$.ts$\in$written) return null@\label{line:canforward}@
	readhead = readhead.next$_0$	@\label{line:lp2_whive}@
	if(readhead.levels==maxlevels)
		inserthead = readhead
	return readhead.winset

updateWindows(tuple, thread) @\label{line:wladdop}@
	windowTSs = getTargetWindowTSs(tuple)
	written$_{thread}$ = windowTSs.first() @\label{line:wstarts}@
	for(wints : windowTSs)
		curnode = inserthead
		for(i = maxlevels-1 downto 0) @\label{line:wupdatestart}@
			next = curnode.next$_i$
			while(next!=tail $\land$ next.ts $\leq$ wints)
				curnode = next
				next = curnode.next$_i$
			update$_i$ = curnode @\label{line:wupdateend}@
		if(curnode.ts != wints)
			winset = new WinSet(wints)
			levels = getLevelHeight()
			newnode = new Node(wints, winset)
			curnode = levelinsert(update$_0$, newnode, wints, 0)
			if(curnode == newnode)
				for(i=1 to levels-1)
					levelinsert(update$_i$, newnode, wints, i)
		win = curnode.winset.find(tuple.key) @\label{line:hashsearch}@
		if(win==null)
			win = new Window()
			curnode.winset.put(tuple.key, win) @\label{line:insertwin}@
		win.processTuple(tuple) @\label{line:processTuple}@

levelinsert(priornode, newnode, wints, level)
	while(true)
		next = priornode.next$_{level}$
		if(level == 0 $\land$ next.ts == wints) return next @\label{line:otheraddht}@
		if(next == tail $\lor$ next.ts > wints)
			newnode.next$_{level}$ = next
			if(CAS(priornode.next$_{level}$, next, newnode) break @\label{line:casLevelInsertWhive}@
		else fromNode = next
	return newnode
\end{lstlisting}
\caption{\winlist}
\label{lst:winsl}
\end{algorithm}
The \winlist (cf. Algorithm~\ref{lst:winsl}) data structure provides lock-free management of \winset{}s.
The \texttt{updateWindows} method adds a tuple to each window it contributes to.
A reference to the earliest such window is saved in the \texttt{written} array for each thread (L\ref{line:wstarts}).
This is used to keep track of when \winset{}s are no longer being updated (L\ref{line:canforward}).
For each window the tuple contributes to, the method traverses the list to locate the \winset with the same timestamp as the window.
This is done in the same manner as when inserting a node into the \tuplist (L\ref{line:wupdatestart}-\ref{line:wupdateend}).
If the \winset is found, it is searched to find the correct window for the tuple's key (L\ref{line:hashsearch}).
If there is no window for the key, a new window is inserted into the \winset with the correct key (L\ref{line:insertwin}).
The tuple is then added to the window.
If no \winset is found for the timestamp, a new \winset and corresponding node
to hold it are created. They are inserted into the list in a similar manner to the \tuplist.
The difference is that another thread might try to create a \winset for the same timestamp concurrently.
If this happens and the other thread manages to insert it, then the insertion must be canceled and the other \winset will be used instead (L\ref{line:otheraddht}).

The \texttt{getNextWinSet} operation returns the next \winset that is no longer being updated by input tuples. It is assumed that it will only be called by a single thread.
If no thread updated any of the windows in the first \winset of the list the last time it received a tuple, it can be assumed that no more tuples will contribute to the \winset in the future, as each thread receives tuples in timestamp order (L\ref{line:canforward}).
If the new head node for the \texttt{getNextWinSet} operation is part of all shortcut levels, it is made the new head node for the \texttt{updateWindows} method.
Nodes and \winset{}s with a timestamp lower than the ones referenced by \texttt{inserthead} and \texttt{readhead} can be safely freed or automatically garbage collected.

Figure~\ref{fig:example_tumlsc} presents how stages \addX{}, \smergeX{}, \updateX{} and \outputX{} (and their respective code lines) are distributed to threads $I_t$ and $O_t$ for the \TUMLSC{} implementation (stages assigned to $I_t$ and $O_t$ threads are colored in blue and red, respectively). Similarly, Figure~\ref{fig:example_tumlmc} shows how such stages are distributed to threads $I_t$, $U_t$ and $O_t$ for the \TUMLMC{} implementation (stages assigned to $I_t$, $U_t$ and $O_t$ threads are colored in blue, red, and yellow, respectively).

\begin{figure}[ht!]
  \centering
  \includegraphics*[width=1\linewidth]{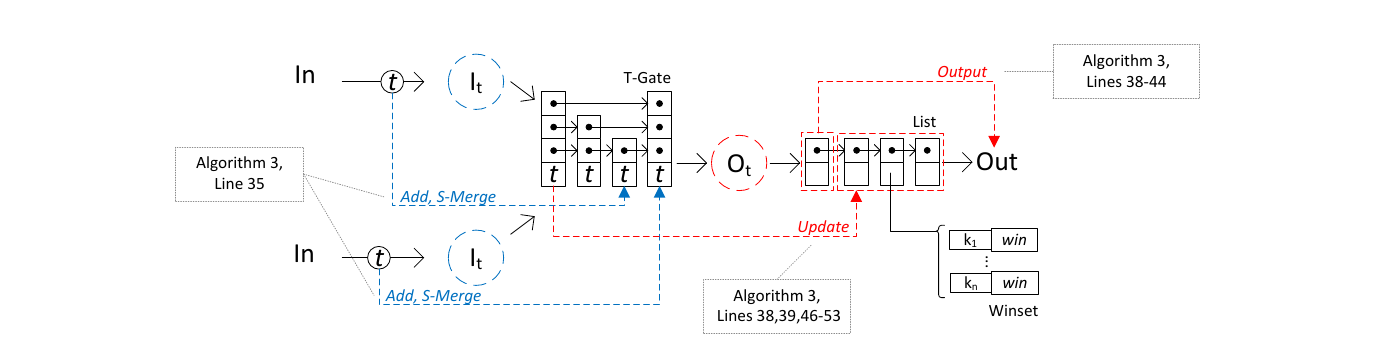}
  \caption{Visual representation of how stages (and their respective code lines) are distributed to threads for the \TUMLSC{} implementation.}
  \label{fig:example_tumlsc}
\end{figure}
\begin{figure}[ht!]
  \centering
   \includegraphics*[width=1\linewidth]{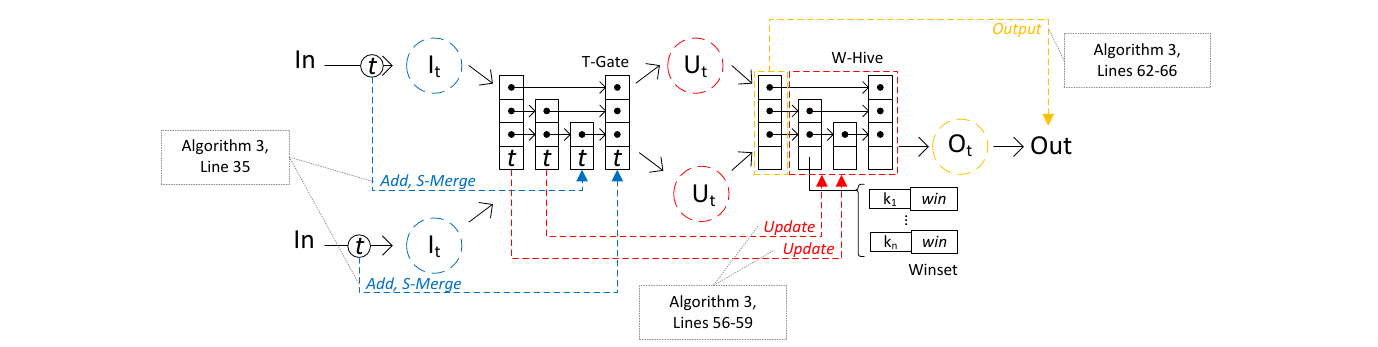}
  \caption{Visual representation of how stages (and their respective code lines) are distributed to threads for the \TUMLMC{} implementation.}
  \label{fig:example_tumlmc}
\end{figure}

\subsection{WiML}
\label{sec:WiML}
\begin{algorithm}
\begin{lstlisting}[name=code, breaklines=true]
AddSMergeUpdate(tuple, input) // One thread per input
	whive.updateWindows(tuple, input)  @\label{line:combaddtowin}@ // See @\lref{lst:winsl}{line:wladdop}@

void Output() // One thread
	winset = whive.getNextWinSet() // See @\lref{lst:winsl}{line:wlreadop}@
	if(winset == null) return
	for(window : winset)
		// Output the result of the window
		forward(window.produceOutTuple()) @\label{line:forwardwin}@
\end{lstlisting}
\caption{\WIML{}}
\label{lst:combined}
\end{algorithm}
The WiML design (see Algorithm~\ref{lst:combined}) is suitable only for aggregate operator's functions $F$ that are 
order-insensitive, since it does not sort the tuples prior to inserting them into their windows.
When a tuple arrives it is immediately processed to update the windows it contributes to.
This is done in the \texttt{AddSMergeUpdate} method using the \winlist (L~\ref{line:combaddtowin}).
The \winlist returns the \winset{}s that will no longer be contributed to, which can then be forwarded (L\ref{line:forwardwin}).

Figure~\ref{fig:example_wiml} presents how stages \addX{}, \smergeX{}, \updateX{} and \outputX{} (and their respective code lines) are distributed to threads $I_t$ and $O_t$ for the \WIML{} implementation (stages assigned to $I_t$ and $O_t$ threads are colored in blue and red, respectively).

\begin{figure}[ht!]
  \centering
  \includegraphics*[width=1\linewidth]{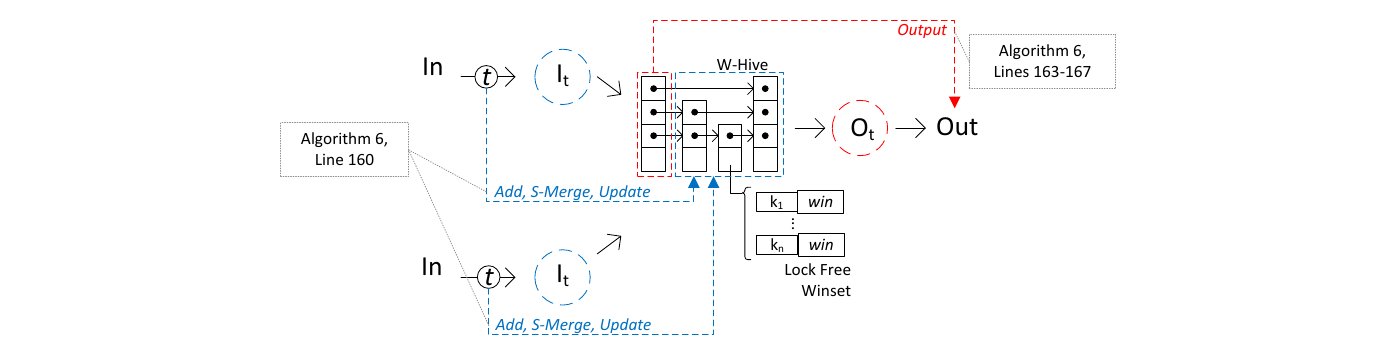}
  \caption{Visual representation of how stages (and their respective code lines) are distributed to threads for the \WIML{} implementation.}
  \label{fig:example_wiml}
\end{figure}

\section{Correctness}\label{sec:correctness}
In this section we outline proofs of liveness and safety properties of the algorithmic constructions of the data structures, namely lock-freedom and linearizability.
\emph{Lock-freedom} guarantees that at least one of the concurrent method call invocations of the data structure will return in a finite number of its own steps~\cite{herlihy_art_2012}.
\emph{Linearizability}~\cite{herlihyandwing} guarantees that every method call appears to take effect at some point (linearization point) between its invocation and response; more formally, for a linearizable implementation of a data structure, given a history of concurrent operations, there exists a sequential ordering of them, consistent with their real-time ordering and with the sequential semantics of the data structure.
Furthermore, we show that the aggregate implementations provide deterministic processing of the stream tuples. 

\begin{customthm}{1}
 The \tuplist implementation presented in algorithm~\ref{lst:tuplesl} is lock-free and linearizable.
\end{customthm}

\begin{proof}
 The \texttt{getNextReadyTuple} method does not contain any loops and returns in a bounded number of its own steps.
 The \texttt{insertTuple} method contains bounded loops except for the  \texttt{levelinsert} subroutine. This will fail to terminate only if the CAS  instruction on L\ref{line:casLevelInsert} fails, i.e. in the case a concurrent  call of \texttt{insertTuple} from another thread makes progress.
 Therefore, the \tuplist implementation is lock-free. 
 
The linearization point of the \texttt{insertTuple} method during concurrent calls  of the same method, is the successful CAS  on L\ref{line:casLevelInsert}, as this is when the operation appears to take effect among such calls.
 
\texttt{getNextReadyTuple} is linearized at the check on L\ref{lp:tgate_getNextReady}
 when the appropriate cell of the \texttt{written} array is read.
 In the case of concurrent calls of methods \texttt{getNextReadyTuple} and \texttt{insertTuple}, 
 the linearization point of the latter is the update of the \texttt{written} array on L\ref{line:written}. 
 Thus there is a linearization point for all the method calls of the \tuplist implementation.
\end{proof}

The \winlist provides management of winsets and is where the actual aggregate computation takes place.
The \texttt{processTuple} call on L\ref{line:processTuple} is an application-specific operation.
Naturally, the safety and liveness properties of the \\\texttt{processTuple} method call affect the ones of the higher level \texttt{update\-Windows} method that includes the former.
Thus, the following theorems are shown under the condition that the \texttt{processTuple} call on L\ref{line:processTuple} is linearizable and lock-free (e.g. local computation in the simplest case).

\begin{customthm}{2}
 The \winlist implementation presented in algorithm~\ref{lst:winsl} is lock-free and linearizable.
\end{customthm}
\begin{proof}
 By definition there are no concurrent calls of \texttt{getNextWinSet}, each such call does not modify any shared variables and returns in a bounded number of its own steps.
 A call to \texttt{update\-Windows} will fail to return only if the CAS instruction on L\ref{line:casLevelInsertWhive} fails (i.e., if a concurrent call from another thread will have made progress).
 In the case of concurrent updates of tuples with the same key, the \texttt{winsets} used are lock-free and linearizable, thus so are all the calls to their methods.
 Therefore, the \winlist implementation is lock-free.
 
 For concurrent calls to the \texttt{updateWindows} method we distinguish two cases: the ones that successfully add a new winset and the respective holding node to the \winlist and the ones that update an existing winset.
 For the latter, the linearization point breaks down to the successful find of the window to be updated in the winset (L\ref{line:hashsearch}), or the insertion of the respective window in case this does not exist (L\ref{line:insertwin}).
 For the former, the linearization point is the successful CAS instruction on L\ref{line:casLevelInsertWhive}.
 A call to \texttt{updateWindows} concurrent with a call to \texttt{getNextWinSet} is linearized on L\ref{line:wstarts}, as this could affect the result of a subsequent check on L\ref{line:canforward} of \texttt{getNextWinSet}.
 The linearization point of \texttt{getNextWinSet} can be any of L\ref{line:lp1_whive} or L\ref{line:canforward} depending on the successful checks, or L\ref{line:lp2_whive} otherwise.
 Thus, there is a linearization point for all the method calls of the \winlist implementation.
 \end{proof}

\begin{mylemma}\label{lem:tgate_ready}
 A tuple $t_i^j$ returned by the \texttt{getNextReadyTuple} method, satisfies the ready definition (cf. Def.~\ref{def:ready}).
\end{mylemma}
\begin{proof}
 Assume towards a contradiction that $t_i^j > merge_{ts}$.
 Then $t_i^j$ would be the tuple with the latest timestamp received by its respective input thread.
 But then the check in line~\ref{lp:tgate_getNextReady} would have failed and $t_i^j$ would not have been returned.
\end{proof}
\begin{mylemma}\label{lem:whive_ready}
 All tuples contributing to a winset returned by the \texttt{getNextWinSet} satisfy the ready definition (cf. Def.~\ref{def:ready}).
\end{mylemma}
\begin{proof}
 Assume there is a tuple that does not.
 As above, the tuple would be the one with the latest timestamp received by its respective input thread.
 In that case the check in line~\ref{line:canforward} would have caused \texttt{null} to be returned, as the winset's timestamp would belong to the windows in the \texttt{written} array.
\end{proof}
 
\begin{customthm}{3}
 The \TUMLSC, \TUMLMC and \WIML aggregate implementations (Alg.\ref{lst:lfcore},\ref{lst:combined}) are lock-free and provide deterministic processing of tuples.
\end{customthm}
\begin{proof}
 All method calls are bounded by a constant and include either local computations for each thread, or calls to data structure implementations that are lock-free (\tuplist,\winlist).
 Thus, the implementations are lock-free.
 
 Lemmas~\ref{lem:tgate_ready} and~\ref{lem:whive_ready} show that the output tuples of the aggregate implementations always consist of ready tuples.
 Therefore, the aggregate implementations are deterministic (cf. Def.~\ref{def:determinism}).
\end{proof}

\section{Evaluation}\label{sec:evaluation}
In this section, we study the performance of the different aggregate operators presented in Section~\ref{sec:newabstractdatatypesandaggregatedesigns} in terms of throughput and latency.
We also include the evaluation of a lock-free version of the multi-queue implementation, referred to as \MQLF{} and relying on the lock-free queue by Michael and Scott~\cite{Michael:1996}, with the sole purpose of showing that \MQ{}'s poor performance does not depend only on the use of locks.
First, we provide evidence of the superiority of the \tuplist{} with respect to a lock-free skip list by measuring their maximum throughput and latency.
Subsequently, we discuss the improvement enabled by \TUMLSC{} and \WIML{}, compared to \MQ{} and \MQLF{}, for different queries and varying number of inputs. Our goal is to show how the \TUMLSC{} and \WIML{} implementations can achieve higher performance than \MQ{} and \MQLF{} ones by increasing the computing time over the synchronization time of their underlying threads. At the same time, we also show the scalability of the \WIML{} implementation for an increasing number of threads running the \updateX{} stage (assigned to $I_t$ threads in this case, as discussed in Section~\ref{sec:WiML}).
In the last part of the section, we evaluate \TUMLMC{}'s scalability for increasing number of threads running the \updateX{} stage, complementing the scalability evaluation of the \WIML{} implementation.
Our experiments take into account the aggregate's features that affect throughput and latency: the overall number of keys, the number of windows to which each tuple contributes and the cost of the aggregate function.
For each feature, we consider 2 stretching points in order to show how traversing its spectrum (e.g., increasing the overall number of keys) affects the overall performance.
All the experiments represent queries that can be found in real-world applications.
We take into account aggregate functions that are commonly used and also evaluate highly costly variants when studying how their cost affects the aggregate performance. Our data sets have been collected from real-world applications.

\subsection{Evaluation setup}

\begin{table*}[!ht]
\scriptsize
\caption{Parameters for queries used in the evaluation. Identifiers \emph{ID} are composed by 3 letters. The first letter (small or capital, representing smaller or larger values/computation-demands) represents the aggregate parameter being studied (k - overall number of keys, w - window size, f - applied function). The last two letters specify whether \emph{F} is order-sensitive (OS) or order-insensitive (OI).}
\label{tab:experiments}
%\tbl{Parameters for queries used in the evaluation. Identifiers \emph{ID} are composed by 3 letters.
%The first letter (small or capital, representing smaller or larger values/computation-demands) represents the aggregate parameter being
%studied (k - overall number of keys, w - window size, f - applied function).
%The last two letters specify whether \emph{F} is order-sensitive (OS) or order-insensitive (OI). \label{tab:experiments}} {
\centering
\begin{tabular}{|c|c|c|c|c|c|p{4cm}|}

  \hline
	ID & \emph{DS} & \emph{WS} & \emph{WA} & \emph{K} & \emph{F} & \emph{Description}\\
  \hline
	\multicolumn{7}{|l|}{Order-sensitive (OS) functions (\MQ{}, \MQLF{}, \TUMLSC{})} \\
  \hline
	k-OS & $EC$ & 30 & 3 & $meter$ & \texttt{first()} & Forward the first consumption reading, group by meter (243 distinct keys) \\
	K-OS & $SC$ & 30 & 3 & $song$ & \texttt{first()} & Forward the first comment, group by song (40,000 distinct keys) \\
	w-OS & $EC$ & 20 & 2 & $meter$ & \texttt{first()} & Forward the first consumption reading, group by meter (each tuple contributes to 10 windows) \\
	W-OS & $EC$ & 40 & 2 & $meter$ & \texttt{first()} & Forward the first consumption reading, group by meter (each tuple contributes to 20 windows) \\
	f-OS & $SC$ & 20 & 2 & $song$ & \texttt{first-mail($cmt$)} & Forward the first comment containing a mail address, group by song \\
	F-OS & $SC$ & 20 & 2 & $song$ & \texttt{first-mail/IP($cmt$)} & Forward the first comment containing a mail or an IP address, group by song \\
  \hline
	\hline
	\multicolumn{7}{|l|}{Order-insensitive (OI) functions (\MQ{}, \MQLF{}, \WIML{})} \\
  \hline
	k-OI & $EC$ & 30 & 3 & $meter$ & \texttt{count()} & Count the number of consumption readings, group by meter (243 distinct keys) \\
	K-OI & $SC$ & 30 & 3 & $song$ & \texttt{count()} & Count the number of comments, group by song (40,000 distinct keys) \\
	w-OI & $EC$ & 20 & 2 & $meter$ & \texttt{avg($cons$)} & Compute the average consumption, group by meter (each tuple contributes to 10 windows) \\
	W-OI & $EC$ & 40 & 2 & $meter$ & \texttt{avg($cons$)} & Compute the average consumption, group by meter (each tuple contributes to 20 windows) \\
	f-OI & $SC$ & 20 & 2 & $song$ & \texttt{count-mail($cmt$)} & Count the number of comments containing a mail address, group by song \\
	F-OI & $SC$ & 20 & 2 & $song$ & \texttt{count-mail/IP($cmt$)} & Count the number of comments containing a mail or an IP address, group by song \\
  \hline
	
\end{tabular}
%}
\end{table*}

The evaluation has been conducted with an Intel-based workstation with two sockets of 6-core Xeon E5645 (Nehalem) processors with Hyper Threading (24 logical cores in total) and 48\,GB DDR3 memory at 1366\,MHz.
The prototype has been implemented in Java and experiments have been run using the OpenJDK Runtime Environment (IcedTea 2.3.9) with the default garbage collection settings.

We use two datasets that we refer to as \textit{SoundCloud} ($SC$) and \textit{Energy Consumption} ($EC$).
$SC$ has been collected from the online audio distribution platform SoundCloud from a subset of approximately $40,000$ users exchanging comments about $250,000$ songs between $2007$ and $2013$.
Tuples contain comments sent by users in relation to songs and are composed by the attributes $\langle ts, user, song, cmt \rangle$.
$EC$ contains energy consumption readings collected from a set of 243 smart meters between May 2012 and June 2013.
Tuples' schema is composed by attributes $\langle ts, meter, cons \rangle$.

Each experiment starts with a warm-up phase and ends with a cool-down phase. During the measuring phase (lasting five minutes) tuples are injected at a constant rate; throughput is measured as the average number of tuples/second (t/s) processed by the aggregate operator during the measuring phase while latency is measured as the average timestamp difference between each output tuple and the latest input tuple that produced it during the measuring phase. Finally, presented results are averaged over $10$ runs.
In each experiment, we deploy one instance of the aggregate, together with injectors, running at dedicated threads and maintaining per-second throughput statistics, and a sink running at a dedicated thread, collecting output tuples and maintaining per-second latency statistics.
When running experiments with different input rates, we process data from the $EC$ and $SC$ datasets, modifying only the rate at which tuples are injected. In order to find the maximum throughput and the corresponding latency of a given setup, several experiments (for increasing input rates) are run, as long as results do not indicate the setup cannot sustain the injected rate.
Table \ref{tab:experiments} presents the parameters of the queries used in the evaluation: identifier \emph{ID}, dataset \emph{DS}, window size \emph{WS}, window advance \emph{WA}, group-by parameter \emph{K} and aggregate function \emph{F}.
For the \textit{\tuplist} and \textit{\winlist}, the maximum possible height of a node ($maxlevels$) is set to $3$ in all experiments.

\subsection{Skip list and \tuplist{} comparison}

In this experiment, we evaluate the performance of a lock-free skip list and the \tuplist{}, measuring the maximum throughput with which $EC$ tuples coming from multiple input streams can be sorted.
For this comparison we used the \\ \texttt{ConcurrentSkipListMap} from Java's \texttt{java.util.concurrent} package, an implementation based also on~\cite{sundell_fast_2005}.

Results for $5,\ 10,\ 15$ and $20$ input streams are presented in Fig.~\ref{fig:comp}.
It can be noted that, when using a skip list, checking for \smergeable{} tuples (done explicitly since the skip list does not differentiates between tuples that are \smergeable{} or not) results in a cost linear to the number of input streams.
Thus, the skip list's throughput degrades while \tuplist{} throughput grows.
For 20 input streams, the skip list is able to sort approximately 1.5 million t/s while the \tuplist reaches approximately 2.2 million t/s (1.5 times better).
The \tuplist also achieves a lower sorting latency, approximately 1 ms for 20 input streams against the 1.6 ms latency of the skip list.

\begin{figure}[h!]
  \centering
  \includegraphics[width=0.24\columnwidth]{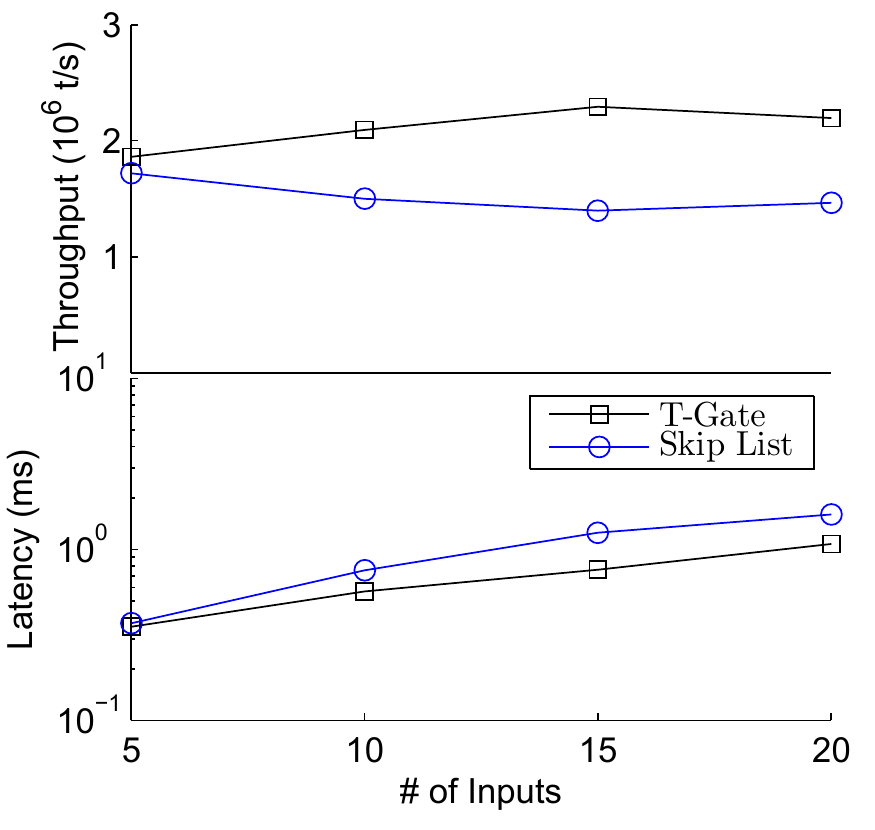}
  \caption{\tuplist{} and Skip List comparison.}
  \label{fig:comp}
\end{figure}

\subsection{Baseline and new designs comparison}

In this set of experiments, we measure the maximum throughput and the latency of the \MQ{}, \MQLF{}, \TUMLSC{} and \WIML{} implementations,
quantifying the improvement enabled by the use of concurrent data structures while using the same number of input and output threads.

\begin{figure}[h!]
\begin{minipage}[b]{0.24\columnwidth}
  \centering
  \includegraphics[width=\columnwidth]{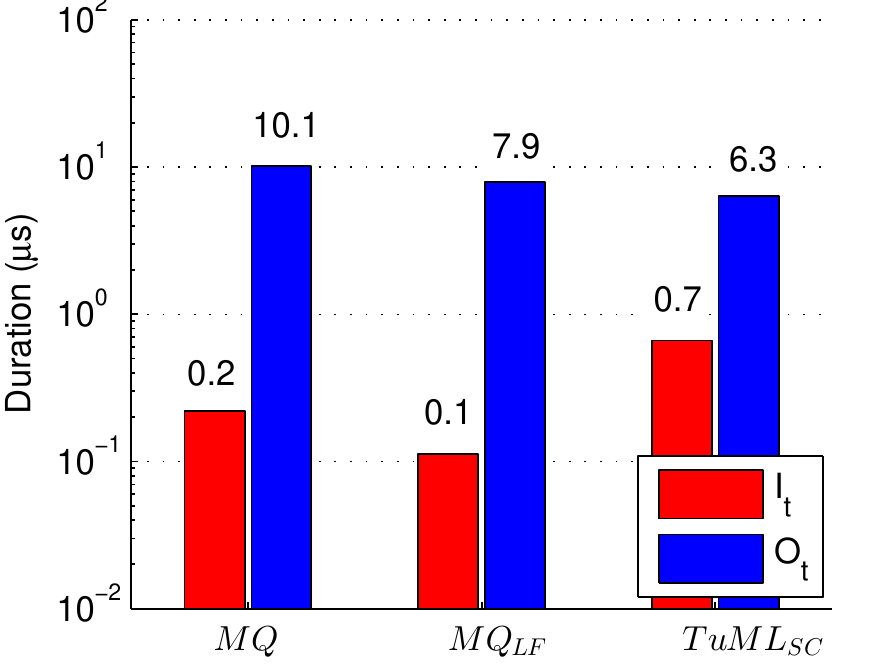}
  \caption{K-OS $I_t$'s and $O_t$'s durations.}
  \label{fig:it_ot_ordersensitive}
\end{minipage}\hfill
\begin{minipage}[b]{0.24\columnwidth}
  \centering
  \includegraphics[width=\columnwidth]{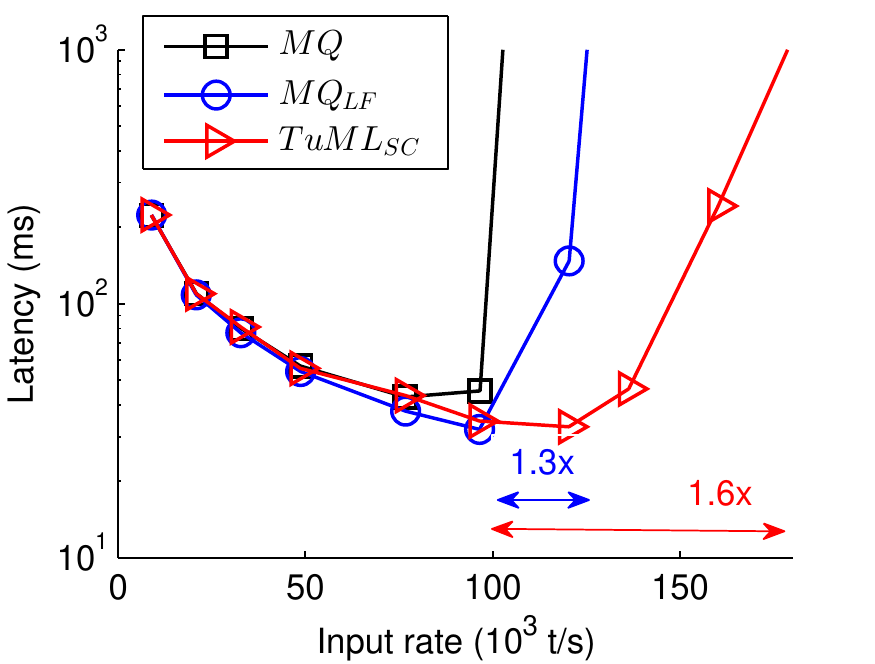}
  \caption{K-OS latency evolution.}
  \label{fig:os_evolution}
\end{minipage}\hfill
\begin{minipage}[b]{0.24\columnwidth}
  \centering
  \includegraphics[width=\columnwidth]{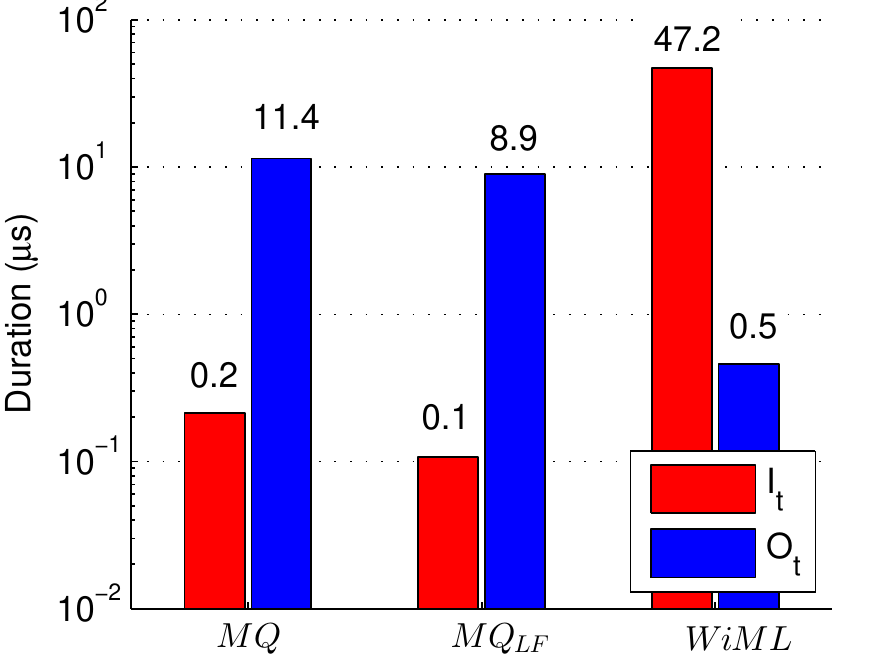}
  \caption{K-OI $I_t$'s and $O_t$'s durations.}
  \label{fig:it_ot_orderinsensitive}
\end{minipage}\hfill
\begin{minipage}[b]{0.24\columnwidth}
  \centering
  \includegraphics[width=\columnwidth]{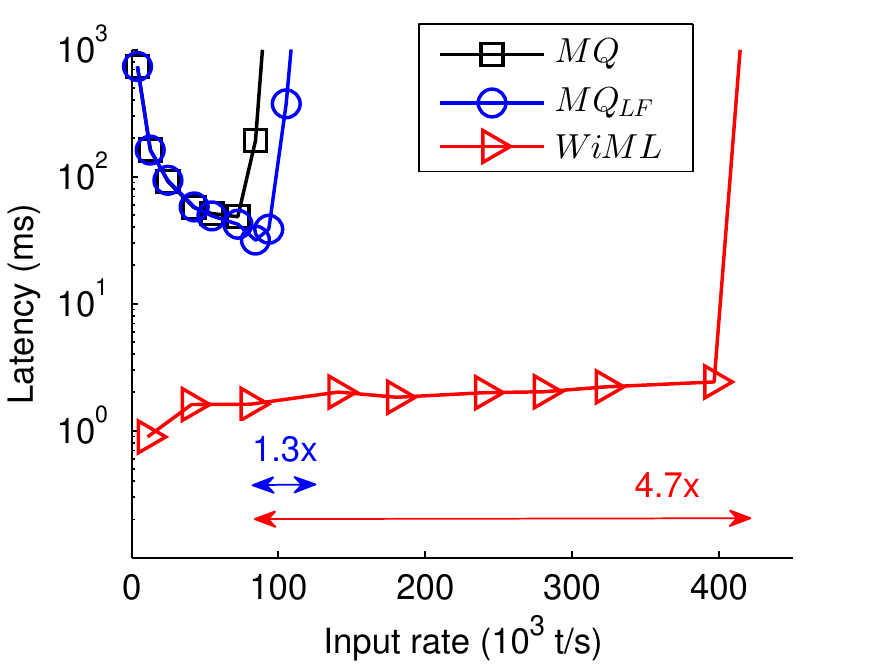}
   \caption{K-OI latency evolution.}
  \label{fig:oi_evolution}
\end{minipage}
  \caption{Throughput and latency improvements enabled by \tuplist{} and \winlist{} in \TUMLSC{} and \WIML{} implementations.}
  \label{fig:t_l_evolution}
\end{figure}

\paragraph*{Parallelization benefit}
We first focus on the average duration of the main operations performed by $I_t$ and $O_t$ threads for queries K-OS and K-OI and 20 input streams.
Results for the query K-OS are presented in Fig.~\ref{fig:it_ot_ordersensitive} (we use logarithmic scale to better appreciate the different orders of magnitude).
Both $I_t$'s and $O_t$'s operations are faster for \MQLF{} compared to \MQ{} since the former relies on lock-free queues.
$I_t$'s duration increases while $O_t$'s decreases for \TUMLSC{} since the \smergeX{} operation is performed by $I_t$.
$O_t$ threads, which constitute the bottleneck, will reach $100,000$, $130,000$ and $160,000$ t/s for \MQ{}, \MQLF{} and \TUMLSC{}, respectively.
Figure~\ref{fig:os_evolution} compares the latency evolution of the different implementations for an increasing input rate.
In all experiments, the latency initially decreases with the increasing rate (lower inter-arrival times at the inputs result in shorter queuing times for input tuples) while it explodes upon saturation of the operator (that is, when the injected load exceeds its maximum throughput).
The throughput achieved by each implementation is close to the expected one (from Fig.~\ref{fig:it_ot_ordersensitive}). \TUMLSC{} achieves a throughput $1.6$ times higher than \MQ{}.
Average durations for query K-OI are shown in Fig.~\ref{fig:it_ot_orderinsensitive}.
It can be noticed that $I_t$'s duration increase is greater for \WIML than \TUMLSC since both \smergeX{} and \updateX{} operations are performed by the $I_t$ thread.
With \WIML{}, each $I_t$ thread is able to process $22,000$ t/s in parallel.
The highest throughput that can be achieved by \WIML{} is in this case given by $O_t$, since the latter can only produce up to $2$ million t/s.
As shown in Fig.~\ref{fig:oi_evolution}, \WIML{} achieves a higher throughput and a lower latency compared to \MQLF{} and \MQ{} (approximately $400,000$ t/s),
independently of the input rate (tuples are processed independently by each $I_t$ thread).

\medskip

In the remaining of this section, we study the performance of \TUMLSC{} and \WIML{} with respect to \MQ{} and \MQLF{} for the different queries presented in Table~\ref{tab:experiments}.
In all experiments, we experience the same performance behavior when comparing \MQ{}, \MQLF{}, \TUMLSC{} and \WIML{} implementations, as explained in the following.
As the number of inputs increases, the latency of multi-queue implementations (\MQ{} and \MQLF{}) increases while their throughput decreases linearly.
This pattern is broken by the \TUMLSC{}, whose throughput is rather stable as the number of inputs increases.
The pattern is even reversed by the \WIML{}, whose throughput actually increases as the number of inputs increases (since stage \updateX{} is run in parallel by each $I_t$ thread) while achieving the lowest and almost constant latency.

\paragraph*{Varying number of keys}
The overall number of keys in the data affects both the operator's throughput and latency since the higher the number of keys, the higher the number of tuples produced for all windows starting at the same timestamp.
Results highlight that \TUMLSC{} and \WIML{} perform better than both \MQ{} and \MQLF{}, whose throughput decreases linearly with the increasing number of inputs.

With respect to order-sensitive functions, we compare queries k-OS (Fig.~\ref{fig:keys_os_ec}) and K-OS (Fig.~\ref{fig:keys_os_sc}).
The query k-OS uses the $EC$ dataset ($243$ distinct keys) while the query K-OS uses the $SC$ dataset ($40,000$ distinct keys).
The upper part of each figure presents the throughput while the bottom part presents the latency (in logarithmic scale).
For 20 inputs, \TUMLSC{} provides the highest throughput ($2.9$ times better than \MQ{}'s for query k-OS) and the lowest latency.
\begin{figure}[h!]
\begin{minipage}[b]{0.24\columnwidth}
  \centering
  \includegraphics[width=\columnwidth]{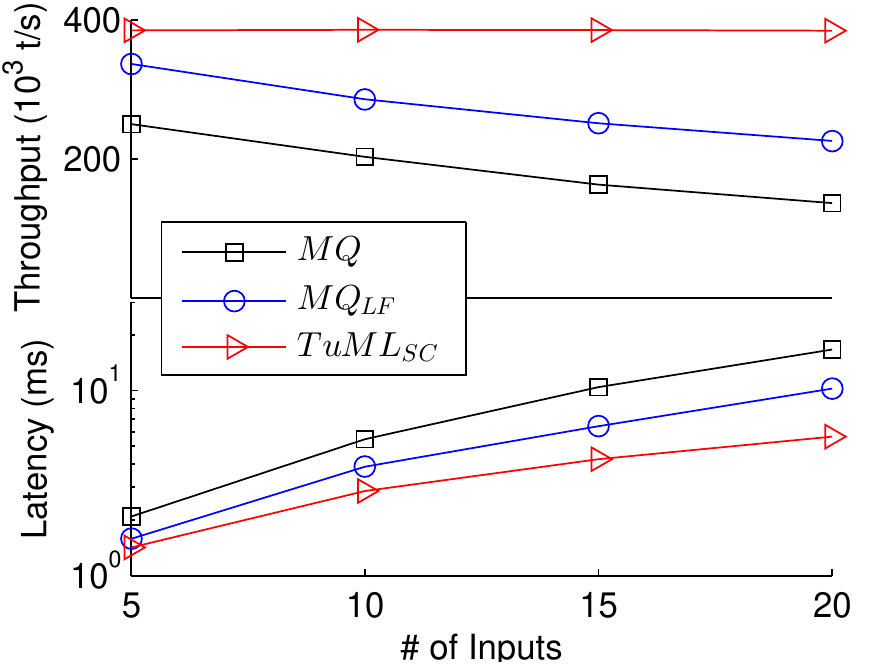}
  \caption{k-OS}
  \label{fig:keys_os_ec}
\end{minipage}\hfill
\begin{minipage}[b]{0.24\columnwidth}
  \centering
  \includegraphics[width=\columnwidth]{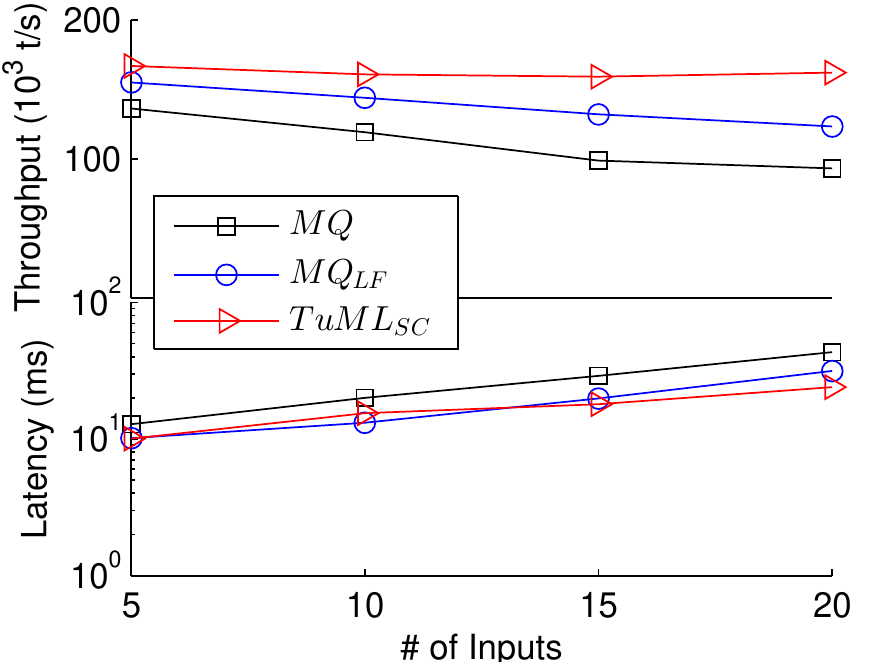}
  \caption{K-OS}
  \label{fig:keys_os_sc}
\end{minipage} \hfill
\begin{minipage}[b]{0.24\columnwidth}
  \centering
  \includegraphics[width=\columnwidth]{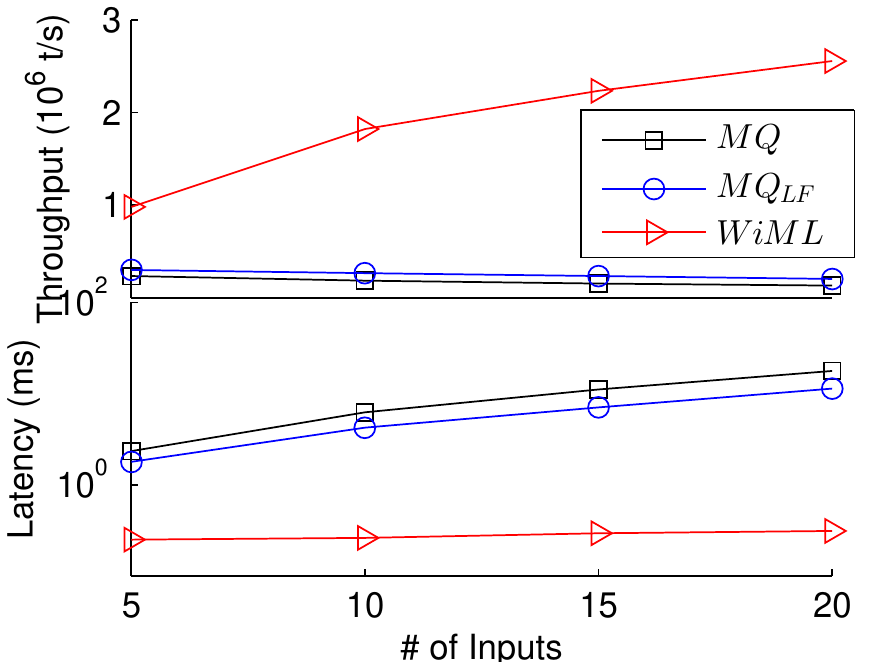}
  \caption{k-OI}
  \label{fig:keys_oi_ec}
\end{minipage}\hfill
\begin{minipage}[b]{0.24\columnwidth}
  \centering
  \includegraphics[width=\columnwidth]{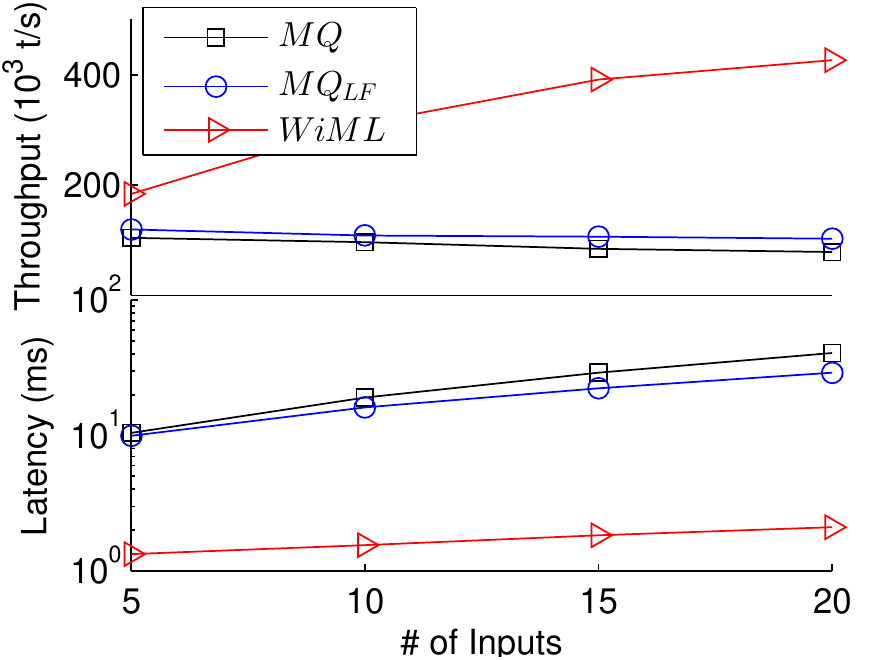}
   \caption{K-OI}
  \label{fig:keys_oi_sc}
\end{minipage}
   \caption{Varying number of keys - evaluation.}
\end{figure}

With respect to order-insensitive functions, we compare queries k-OI (Fig.~\ref{fig:keys_oi_ec}) and K-OI (Fig.~\ref{fig:keys_oi_sc}).
As for order-sensitive functions, both \MQ{}'s and \MQLF{}'s throughput decreases for increasing number of inputs.
On the other hand, \WIML{} throughput increases accordingly to the number of inputs, achieving a maximum throughput of $2.6$ million t/s and $430,000$ t/s, respectively.
\MQ{}'s and \MQLF{}'s latencies increase with the number of inputs while \WIML{}'s one remains approximately constant.

\paragraph*{Varying number of windows to which tuples contribute}
The rationale for this experiment is that the higher the number of windows to which each input tuple contributes, the higher the duration of the \updateX{} operation.
Also in this case, our enhanced implementations outperform both \MQ{} and \MQLF{}.
An increasing number of windows to which tuples contribute results in an overall throughput breakdown and latency increase.

We first focus on order-sensitive functions with queries w-OS (Fig.~\ref{fig:wins_OS_20}) and W-OS (Figure~\ref{fig:wins_OS_40}).
For query w-OS, each tuple contributes to 10 windows.
For query W-OS, each tuple contributes to 20 windows.

\begin{figure}[h!]
\begin{minipage}[b]{0.24\columnwidth}
  \centering
  \includegraphics[width=\columnwidth]{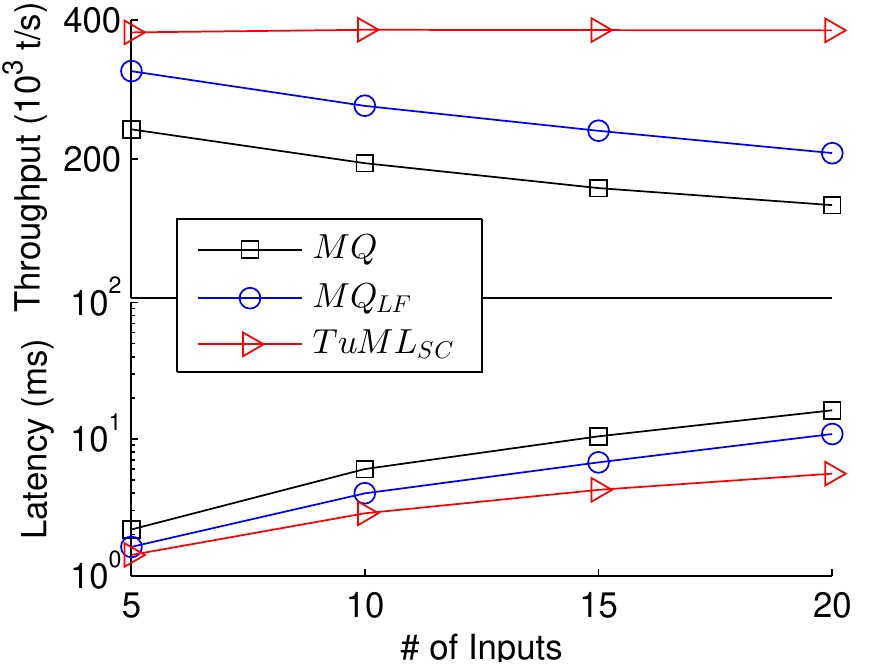}
  \caption{w-OS}
  \label{fig:wins_OS_20}
\end{minipage}\hfill
\begin{minipage}[b]{0.24\columnwidth}
  \centering
  \includegraphics[width=\columnwidth]{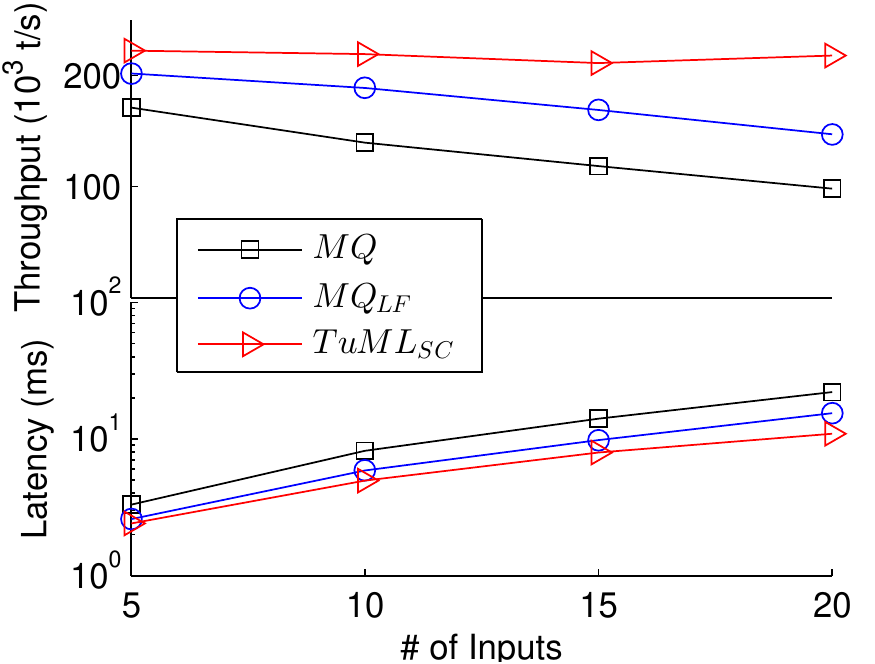}
  \caption{W-OS}
  \label{fig:wins_OS_40}
\end{minipage} \hfill
\begin{minipage}[b]{0.24\columnwidth}
  \centering
  \includegraphics[width=\columnwidth]{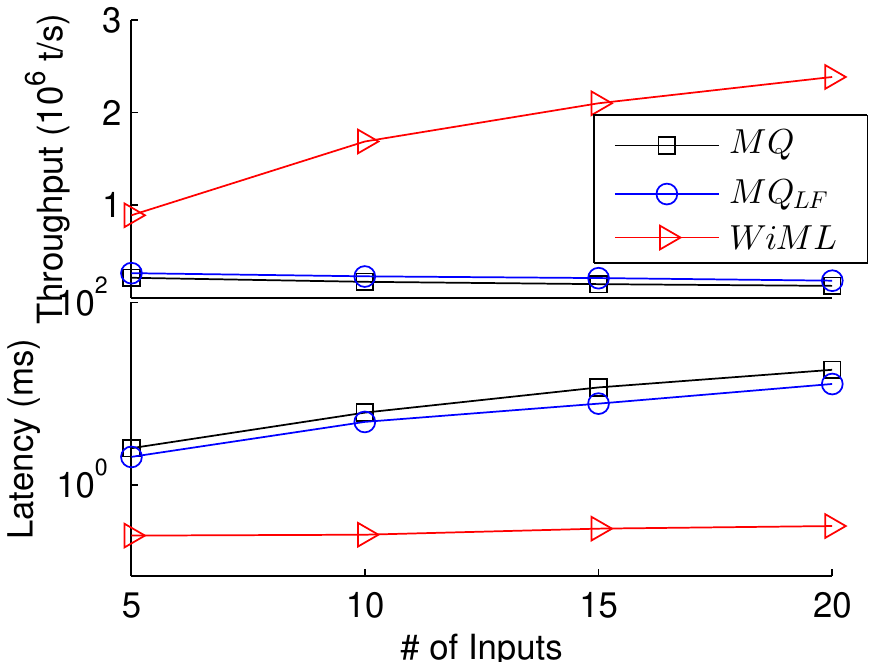}
  \caption{w-OI}
  \label{fig:wins_OI_20}
\end{minipage}\hfill
\begin{minipage}[b]{0.24\columnwidth}
  \centering
  \includegraphics[width=\columnwidth]{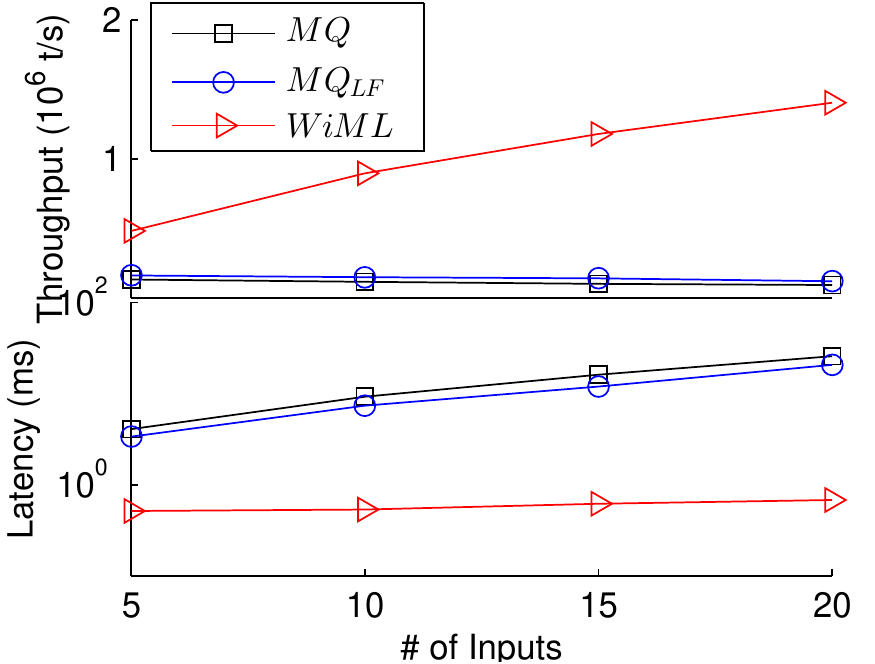}
   \caption{W-OI}
  \label{fig:wins_OI_40}
\end{minipage}
   \caption{Varying number of windows to which tuples contribute - evaluation.}
\end{figure}

With respect to order-insensitive functions, we compare queries w-OI (Fig.~\ref{fig:wins_OI_20}) and W-OI (Fig.~\ref{fig:wins_OI_40}).
Also in this case, \WIML{} outperforms \MQ{} and \MQLF{}.
For queries w-OI and W-OI, \WIML{}'s maximum throughput is of approximately $2.4$ and $1.4$ million t/s, $19$ and $16$ times better than \MQ{}, respectively.

\paragraph*{Varying function cost}
In this set of experiments, we study how throughput and latency evolve with respect to different function costs. We expect the throughput to decrease and the latency to increase accordingly to the increasing cost of the aggregation function.
As observed before, \TUMLSC{} performs better than multi-queue implementations, although this improvement becomes smaller when running very expensive aggregate functions.
The reason for this smaller improvement is due to the dominance of the heavy aggregate function computations over other computations performed by the operator (e.g., the sorting ones) given that both  \TUMLSC{} and multi-queue implementations define a single thread dedicated to the former.
\WIML{} outperforms multi-queue implementations both in terms of throughput and latency independently of the aggregate function cost (in this case, despite relying on the same overall number of threads, the expensive aggregate function is run in parallel by each $I_t$ thread).

With respect to order-sensitive functions, we compare queries f-OS (Fig.~\ref{fig:re1_os}) and F-OS (Fig.~\ref{fig:re2_os}).
When increasing the function cost (query F-OS), \TUMLSC{}'s throughput and latency become really close to \MQ{}'s and \MQLF{}'s ones.
\begin{figure}[h!]
\begin{minipage}[b]{0.24\columnwidth}
  \centering
  \includegraphics[width=\columnwidth]{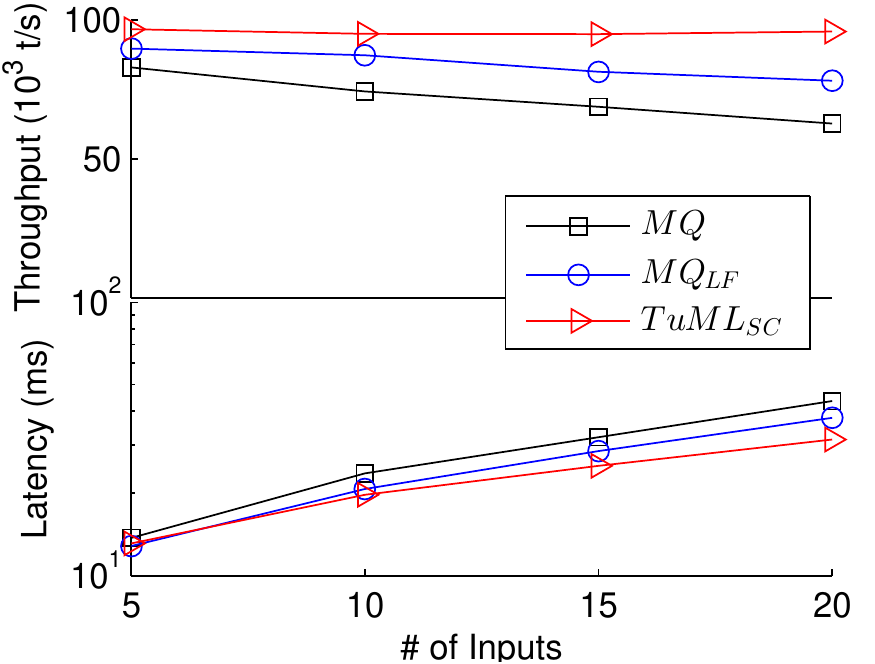}
  \caption{f-OS}
  \label{fig:re1_os}
\end{minipage}\hfill
\begin{minipage}[b]{0.24\columnwidth}
  \centering
  \includegraphics[width=\columnwidth]{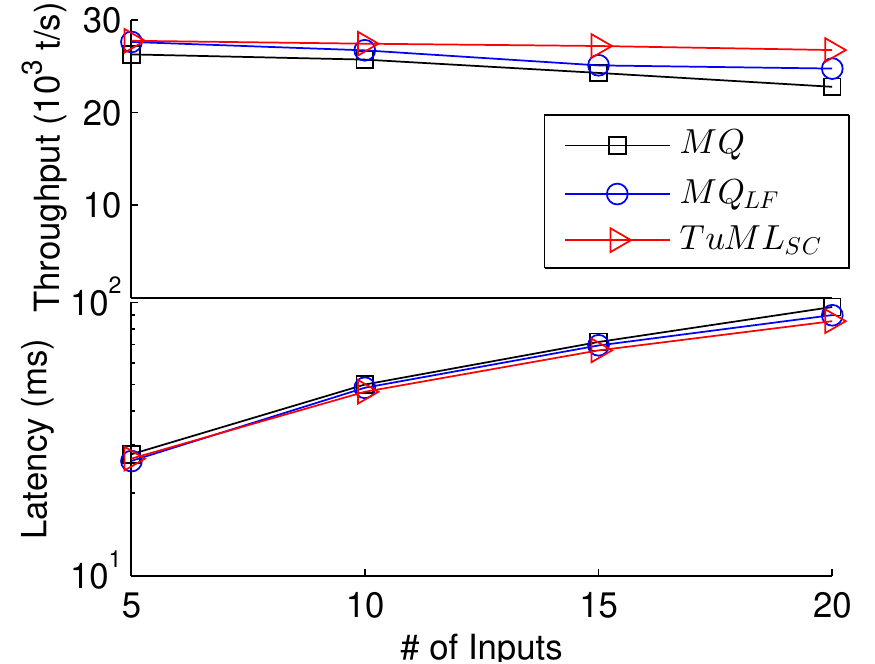}
  \caption{F-OS}
  \label{fig:re2_os}
\end{minipage} \hfill
\begin{minipage}[b]{0.24\columnwidth}
  \centering
  \includegraphics[width=\columnwidth]{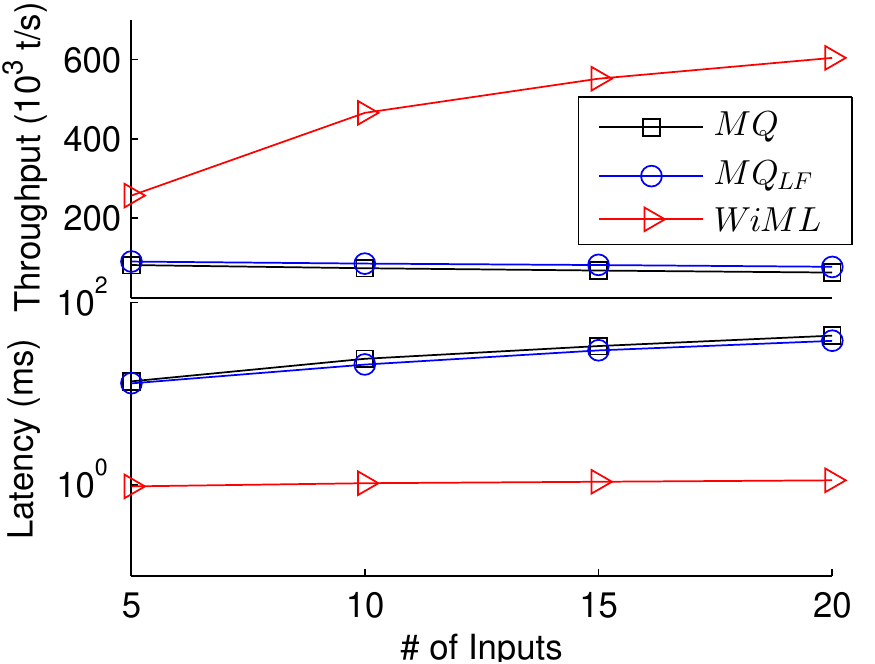}
  \caption{f-OI}
  \label{fig:re1_oi}
\end{minipage}\hfill
\begin{minipage}[b]{0.24\columnwidth}
  \centering
  \includegraphics[width=\columnwidth]{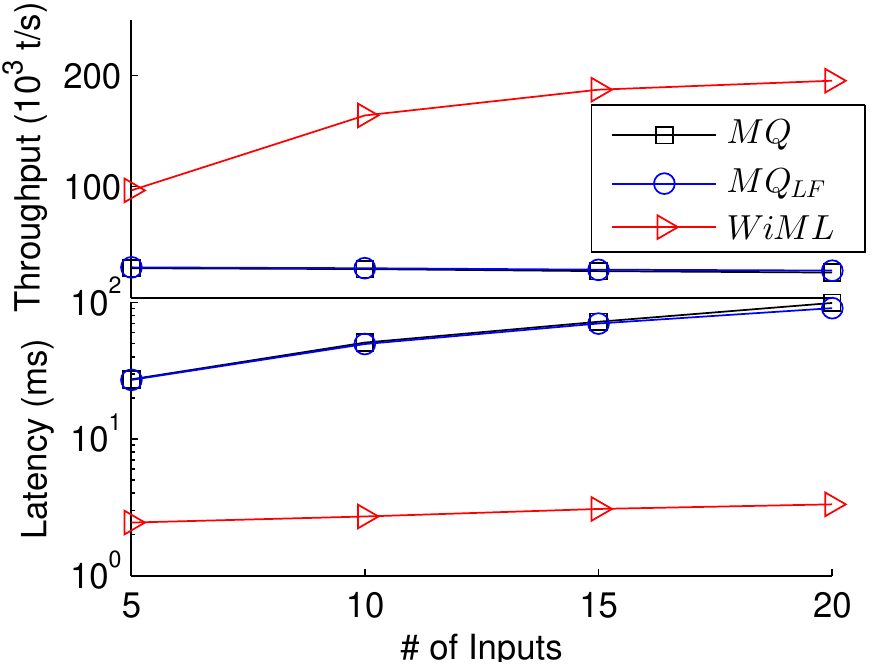}
   \caption{F-OI}
  \label{fig:re2_oi}
\end{minipage}
   \caption{Varying function cost - evaluation.}
\end{figure}
Throughput and latency evolution for order-insensitive functions are evaluated for queries f-OI (\ref{fig:re1_oi}) and F-OI (\ref{fig:re2_oi}).
For both experiments, \WIML{} achieves a throughput of approximately $615,000$ t/s, $9$ times better than \MQ{}.

\subsection{\TUMLMC{} scalability evaluation}\label{sec:sl}

In this section, we focus on the scalability of the \TUMLMC{} implementation.
We execute all the previous queries for order-sensitive functions using the \TUMLMC{} implementation for an increasing number of threads (up to $12$, the physical number of cores of the machine used in the evaluation).
For each query, we present how the throughput and the latency evolve when considering $5$ and $20$ inputs streams.

\begin{figure}[h!]
\begin{minipage}[b]{0.24\columnwidth}
  \centering
  \includegraphics[width=\columnwidth]{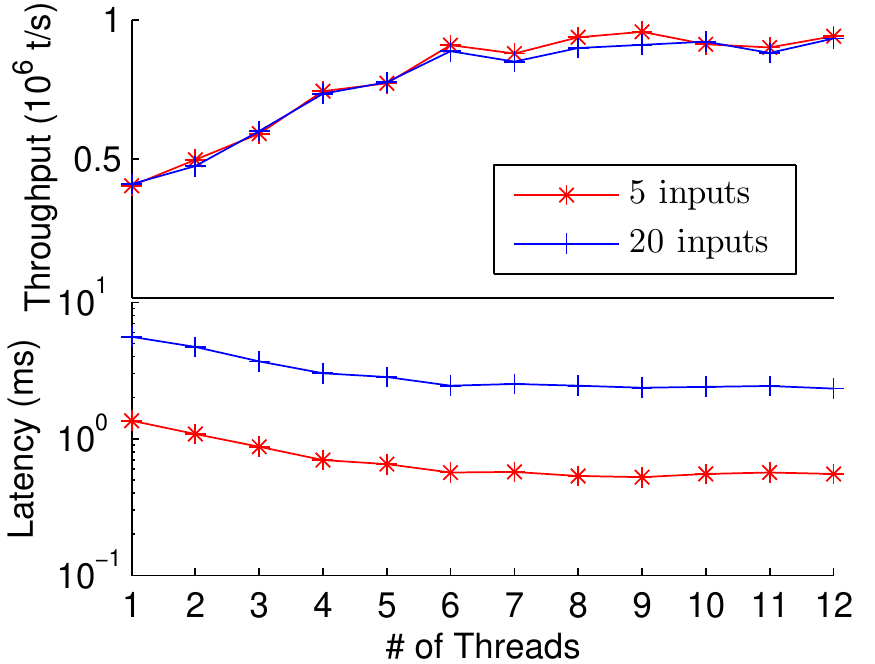}
  \caption{k-OS}
  \label{fig:k1_sl}
\end{minipage}\hfill
\begin{minipage}[b]{0.24\columnwidth}
  \centering
  \includegraphics[width=\columnwidth]{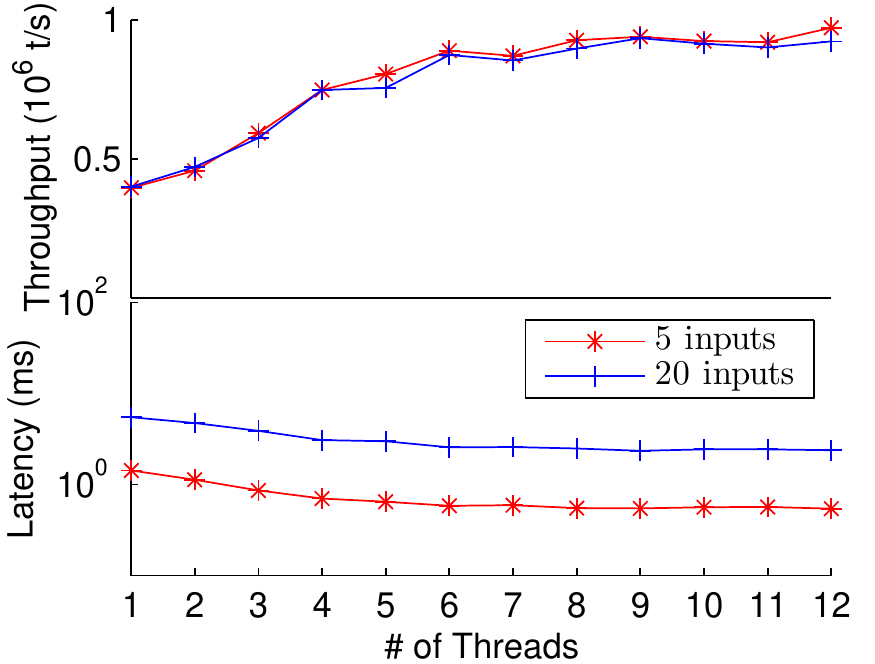}
  \caption{w-OS}
  \label{fig:w1_sl}
\end{minipage}\hfill
 \begin{minipage}[b]{0.24\columnwidth}
  \centering
  \includegraphics[width=\columnwidth]{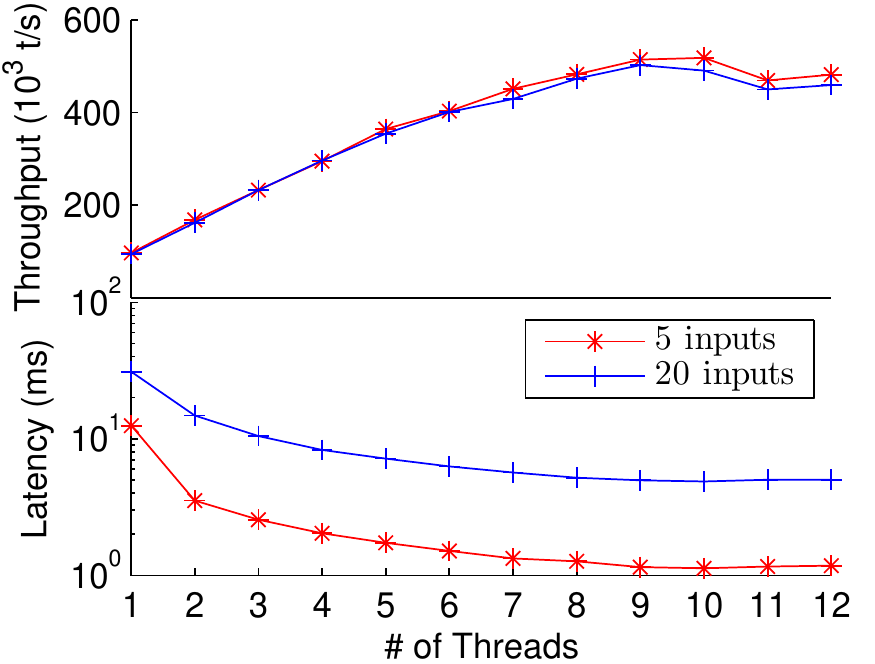}
  \caption{f-OS}
  \label{fig:f1_sl}
\end{minipage}\\
\begin{minipage}[b]{0.24\columnwidth}
  \centering
  \includegraphics[width=\columnwidth]{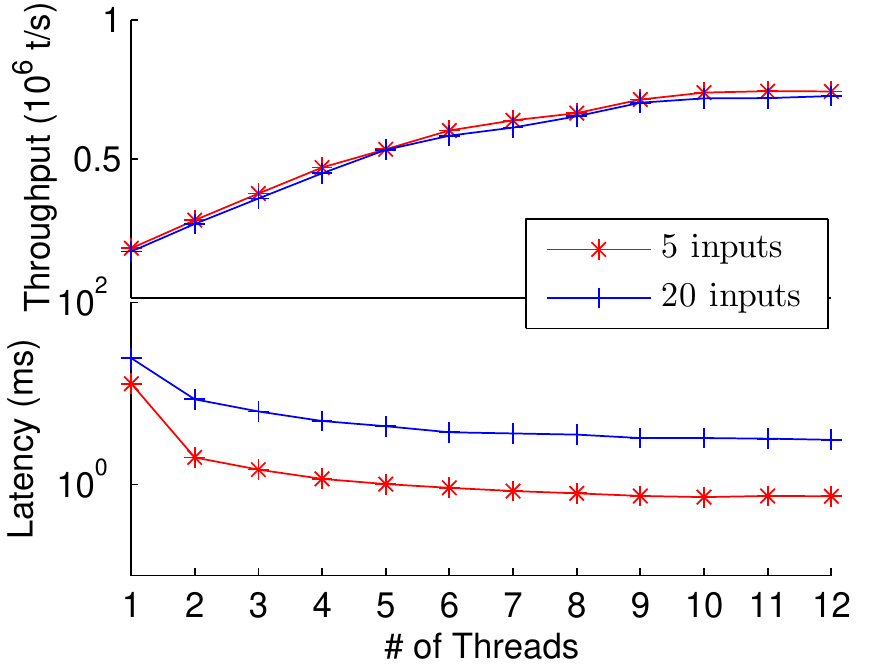}
  \caption{K-OS }
  \label{fig:k2_sl}
\end{minipage} \hfill
\begin{minipage}[b]{0.24\columnwidth}
  \centering
  \includegraphics[width=\columnwidth]{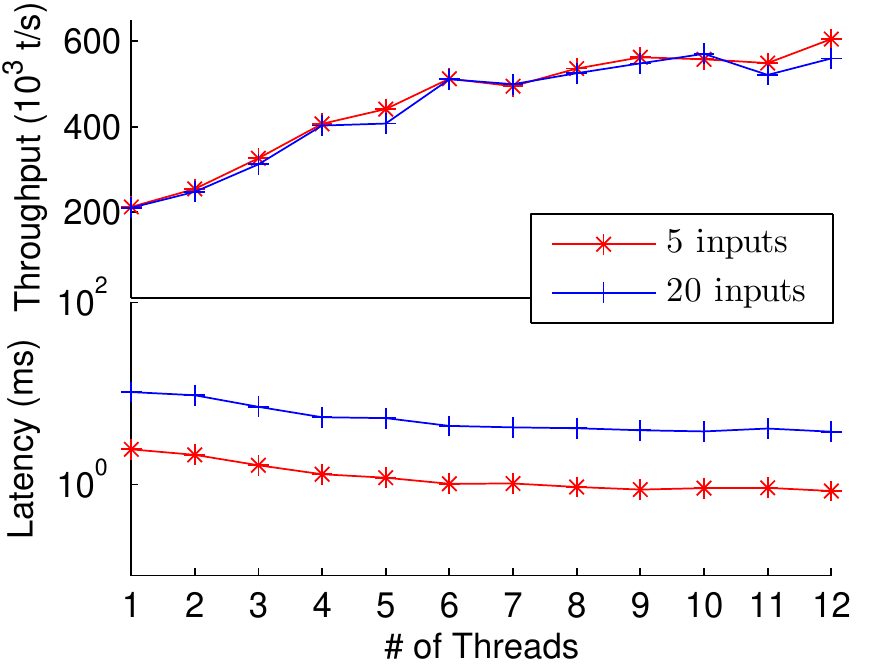}
  \caption{W-OS}
  \label{fig:w2_sl}
\end{minipage}\hfill
\begin{minipage}[b]{0.24\columnwidth}
  \centering
  \includegraphics[width=\columnwidth]{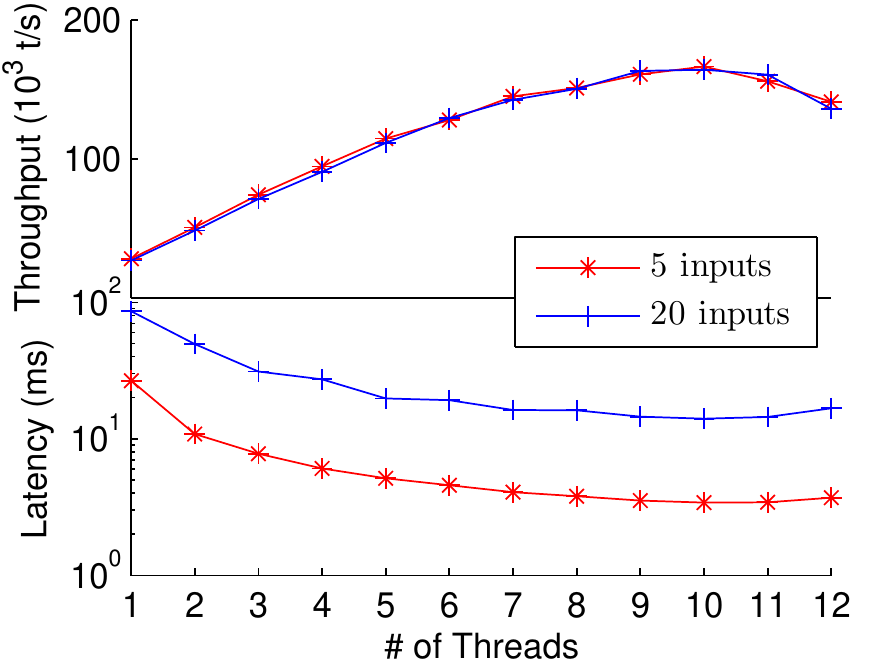}
  \caption{F-OS}
  \label{fig:f2_sl}
\end{minipage}
   \caption{\TUMLMC{} scaling for increasing number of threads.}
   \label{fig:sl}
\end{figure}
Figures~\ref{fig:k1_sl} and \ref{fig:k2_sl} present the throughput and latency evolution for queries k-OS and K-OS.
It can be noticed that K-OS scales better than k-OS for an increasing number of threads due to the increase of $O_t$ operations' duration caused by the higher number of keys.
Figures~\ref{fig:w1_sl} and \ref{fig:w2_sl} present the throughput and latency evolution for queries w-OS and W-OS.
In this experiment, throughput and latency behave similarly despite the increased duration of $O_t$ operations. This is because the increased $O_t$ operations' duration is actually caused by an higher number of windows updated by each tuple, resulting in an higher contention in the underlying \winlist.
Finally, Figures~\ref{fig:f1_sl} and \ref{fig:f2_sl} present the throughput and latency evolution for queries f-OS and F-OS.

\subsection{Summary of results}
Comparing the implementations that rely on one $I_t$ thread per input and a single output thread $O_t$, both \TUMLSC{} and \WIML{} perform better than \MQ{} and \MQLF{}, enabling coping with streams of higher speed.
The improvement enabled by \TUMLSC{} is more sensitive to the aggregate parameters than \WIML{}, which clearly outperforms \MQ{} and \MQLF{}.
When increasing the number of processing threads, \TUMLMC{}'s performance increases both in terms of throughput and latency.
Moreover, its scaling does not degrade when increasing the number of threads above the number for which the highest rate is achieved.

\section{Related Work}\label{sec:relatedwork}

Parallel execution of data streaming operators has been addressed mainly by means of partitioned parallelism \cite{gulisano2012streamcloud,balkesen2013adaptive}, where multiple instances of an operator are assigned to distinct partitions of a given stream.
The way tuples are routed to instances (round-robin, hash-based \cite{gulisano2012streamcloud} or pane-based \cite{balkesen2013adaptive}) depends on the operator's semantics.
It should be noticed that partitioned parallelism is orthogonal to our parallelization technique since we focus on the performance improvement of individual instances of an operator.
The work presented in \cite{schneider2009elastic} discusses a multithreaded elastic streaming protocol that adjusts the number of processing threads depending on the system load. Similarly to our \TUMLMC{} implementation, the protocol defines a single \textit{work queue} from which multiple \textit{worker threads} consume tuples. Nevertheless, that does not take into account sorting of input tuples, which is one of the key challenges addressed in our work. Moreover, the authors do not discuss improvements enabled by concurrent data structures in the multithreaded environment.
With respect to parallel data streaming in the context of multi-core CPUs and GPUs, \cite{schneidert2010evaluation} present a parallel implementation of the aggregate operator and study how it performs on distinct parallel architectures.
The aggregate model discussed by the authors differs from ours since windows are tuple-based and the overall number of distinct group-by values is known before-hand and does not vary over time. Moreover, no discussion is provided about deterministic processing in the context of multiple input sources.
Parallel processing in multi-core CPUs and GPUs is also discussed by \cite{cugola2012low} but, differently from us, the authors focus on pattern detection rather than data aggregation and rely on automata-based incremental processing.

As discussed in section \ref{sec:datastreamingandmultiwayaggregation}, one of the challenges in providing deterministic processing is the merging of multiple timestamp sorted input streams.
This has been discussed in the context of parallel-distributed SPEs \cite{srivastava2004flexible,gulisano2012streamcloud} and replica-based fault tolerance protocols for data streaming \cite{balazinska2008fault}.
Existing approaches for streaming aggregation rely on separated input queues (similar to the $MQ$ protocol).
As shown in our evaluation, this merging is not efficient and implementations such as \TUMLSC{} can drastically improve the overall operator's performance.

\section{Conclusions and Future Work}\label{sec:conclusions}

Providing the appropriate data structures that best fit the needs of a concurrent application is a key research issue, as emphasized by \cite{Shavit:2011,Michael:2013} and also seen in examples such as \cite{GidenstamPT:2010,wimmer:2014}.
In this paper, we study data structures as articulation points in  the context of streaming aggregation, analysing the concurrency needs and proposing methods to meet them.
We propose proper data structures for managing tuples and windows (\tuplist{} and \winlist{}).
Their operations and their lock-free implementations
enable better interleaving and hence improve the balancing and the parallelism of the aggregate operator's processing stages.
As shown in the extensive evaluation based on real-world datasets, our enhanced aggregate implementations outperform existing ones both in terms of throughput and latency, and are able to handle heavier streams, increasing the processing capacity up to one order of magnitude.

These results and the analysis of the role of data structures as articulation points to facilitate concurrency and balancing of the work among streaming operators, open up new venues in the broader context of data streaming, including the enhancement of other operators and the enhancement of SPEs architectures.

%\bibliography{LFDS}

\bibliographystyle{plain}

\end{document}